%% file: paper_draft.tex
\documentclass{article}
\usepackage{amsopn}
\usepackage[english]{babel}
\usepackage[margin=1.in]{geometry}
\usepackage{amsthm}
\input{preambule.tex}
\input{commands.tex}
\definecolor{matchaColor}{HTML}{74A12E} % arbitrary hex code, no #
\newcommand{\inprod}[2]{\left\langle #1,\,#2\right\rangle}

\newcommand{\matcha}{\includegraphics[height=1em]{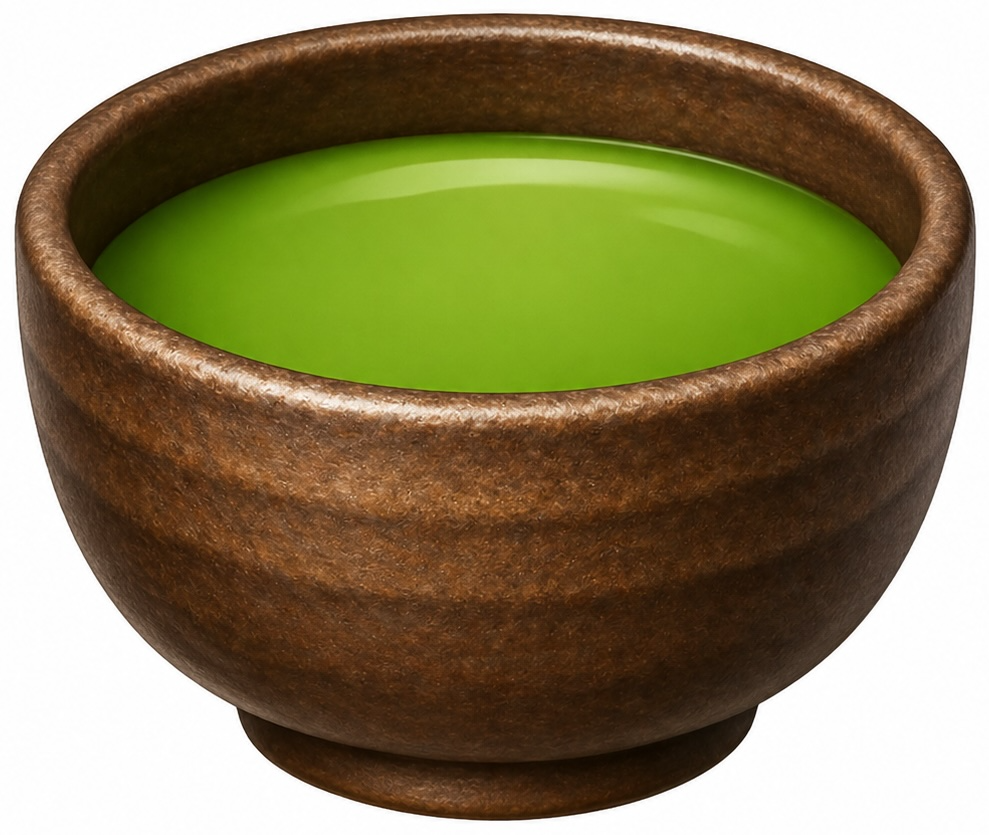}}
\title{\matcha{} Fast Volume Alignment by \\Frequency-Marched Newton Method}
\date{}
\author{Fabian Kruse\thanks{Department of Mathematics and Computer Science, University of Basel, Basel, 4051, Switzerland}
\and Valentin Debarnot\thanks{INSA‐Lyon, Universit\'e Claude Bernard Lyon 1, CNRS, Inserm, CREATIS UMR 5220, U1294, Lyon, France}
\and Vinith Kishore\footnotemark[1]
\and Ivan Dokmanić\footnotemark[1]}

\def\method{Matcha}

\begin{document}

\maketitle

\begin{abstract}
We introduce \emph{Matcha}, a fast method for rotational pose estimation in three-dimensional alignment, and combine it with FFT-based translation updates for full pose estimation. Classical matched filtering evaluates cross-correlation over a large discretized transformation space; we instead treat rotational alignment as a continuous optimization problem on \(\SO(3)\). Matcha starts from a bandlimited Wigner--\(D\) expansion of the rotational correlation, which enables rapid objective evaluation together with analytic gradients and Hessians.
A low-bandwidth SOFFT search provides robust candidate rotations, which are then refined by frequency marching: the angular bandwidth is progressively increased, and candidates are updated by Newton steps at each level. This confines exhaustive search to a single low-frequency stage while allowing the final accuracy to be determined by continuous refinement rather than by the grid spacing.
We prove a deterministic conditional guarantee showing that, under reasonable assumptions, Matcha returns a near-optimal solution for the final bandlimited objective. On synthetic rotation-estimation benchmarks, Matcha attains sub-degree accuracy while substantially reducing runtime relative to exhaustive \(\SO(3)\) search. Integrated into a RELION-5 subtomogram-averaging workflow, it matches the baseline reconstruction quality on the tested dataset, reaching the same Nyquist-limited local resolution while reducing rotational pose-refinement time by more than an order of magnitude.
\end{abstract}

\section{Introduction}
We consider the problem of estimating the pose--rotation and translation--of a particle in a three-dimensional volume, under severe noise and experimental artifacts. Our primary motivation comes from subtomogram averaging (STA) in cryogenic electron tomography (cryo-ET), where highly accurate particle alignment is essential to achieve high resolution. In practice, alignment is performed for every particle and is typically repeated across several refinement iterations, which makes even a modest per-particle cost compound and motivates the search for methods that are at once accurate and computationally efficient.

A classic approach to this problem is \emph{matched filtering}, which finds the transformation that maximizes the correlation between a reference template and the observed data.
Matched filtering is optimal for detecting a deterministic signal corrupted by additive white Gaussian noise, in the sense that it maximizes the output signal-to-noise ratio \cite{turin_matched_1960,kay_detection_1998,debarnot2023blind}.

Letting $g\in\Gc$ denote the unknown transformation to recover, the matched filtering objective is the correlation
\begin{equation}\label{eq:cc}
    \CC(g) \eqdef \left\langle f ,g \circ h\right\rangle = \int_{\mathbb{R}^d} f(x)\,\overline{h(g^{-1}x)}\,dx,
\end{equation}
where $f,h:\mathbb{R}^d\to\mathbb{R}$ denote the noisy measurement and the reference template; we consider $d=3$.
Pose estimation then solves
\begin{equation}\label{eq:maxcc}
    g^\star \in \argmax_{g\in\Gc} \CC(g).
\end{equation}

A straightforward strategy is to evaluate $\CC$ over a suitable discretization of $\Gc$. This may be effective for a coarse search over low-dimensional pose spaces, but already for 3D rotations, a grid with angular spacing \(\Delta\) requires on the order of \(1/\Delta^{3}\) samples; including translations increases the exponent to 6. Correlation can sometimes be evaluated more efficiently on a structured
grid using fast Fourier transforms on translations and rotations via Wigner/\(\SO(3)\) FFTs~\cite{kostelec_ffts_2008,price_differentiable_2024}.
For a pure grid-search strategy, however, the returned pose is limited by the
grid spacing.
An alternative is to treat pose estimation as a \emph{continuous optimization problem} over $\Gc$.
When $\Gc$ has a differentiable structure--such as $\SO(3)$ or $\mathbb{R}^3$--the correlation function admits well-defined gradients and Hessians, enabling efficient gradient-based and second-order optimization methods.
Direct second-order optimization of the full high-resolution rotational correlation is less common in large-scale STA workflows: the correlation landscape is highly non-convex with many spurious local maxima, and it becomes increasingly ill-conditioned at high resolution, in the sense that the local curvature near prominent maxima becomes highly anisotropic (i.e., the Hessian has a large condition number), making second-order steps sensitive to noise and numerical error. Success thus depends on effective initialization.

To address this challenge, we exploit the fact that smoother input volumes induce smoother correlation landscapes. Bandlimiting the inputs to angular degree $L$ produces a bandlimited correlation objective $\CC_L$. Lower $L$ yields a smoother landscape and, in typical imaging problems, fewer prominent local maxima and lower curvature, so the basin of attraction of the global maximizer can be identified via a coarse search at moderate cost. 
The theoretical and algorithmic core of the paper concerns rotational alignment. In the full subtomogram-averaging pipeline, we combine this fast rotational refinement with standard FFT-based translation updates.
We call the resulting procedure \method{}. \method{} first performs a coarse exhaustive search at low bandwidth $L_0$ to identify a set of candidate rotations. The bandwidth is then progressively increased through a schedule $L_0 < L_1 < \cdots < L_J$ (this is sometimes called \emph{frequency marching}). The key point is that at each stage the candidates are refined via Newton steps on $\CC_{L_j}$, initialized from the previous level. Given the favorable regularity of bandlimited functions, the Newton method converges rapidly.
At low and intermediate bandwidths, the correlation landscape is smooth and well-conditioned, and typically only $1$--$2$ Newton steps per bandwidth suffice (Section~\ref{sec:experiments}).
Because $\Gc$ is low-dimensional (three-dimensional for $\SO(3)$), the Hessian is a $3\times 3$ matrix, so solving the Newton system has negligible cost compared with evaluating the Wigner--\(D\) expansion. By confining exhaustive search to a single low-frequency stage and using continuous optimization for refinement, this approach achieves sub-degree angular precision at substantially reduced computational cost compared with exhaustive high-resolution search. We emphasize that the final accuracy is determined by the continuous refinement, not by the resolution of the initial grid, so \method{} can achieve precision far beyond what the coarse SOFFT grid alone would permit.

Figure~\ref{fig:illustration} gives a schematic one-dimensional illustration of the principle behind \method{}. At low frequencies, the correlation landscape is smooth and contains few local maxima, enabling reliable coarse search. As higher frequencies are introduced, the landscape becomes increasingly oscillatory. By tracking dominant maxima across bandwidths and refining them via local optimization, the method aims to recover the optimal pose without exhaustive search at full resolution.

\begin{figure}
	\centering
	    \begin{tikzpicture}[spy using outlines={circle,orange,magnification=4,size=1.5cm, connect spies}]
            \node[rotate=0, line width=0.05mm, draw=white] at (0,0) {\includegraphics[height=6cm]{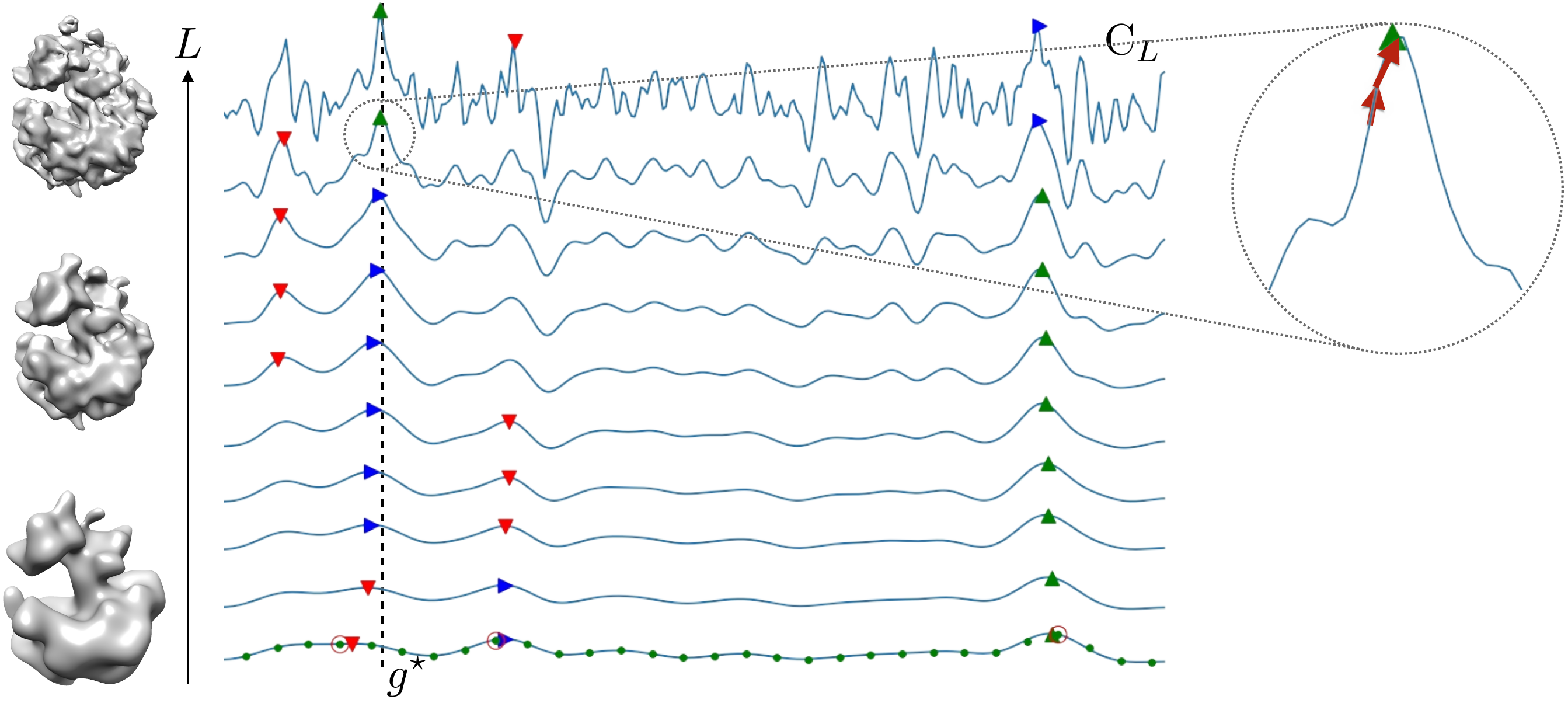}};
    		  \node[rotate=0, line width=0.05mm, draw=white, anchor=west] at (3.5,-1.3) {{\footnotesize{\textcolor{darkgreen}{$\blacktriangle$} First local maximum}}};
            \node[rotate=0, line width=0.05mm, draw=white, anchor=west] at (3.5,-1.8) {{\footnotesize{\textcolor{blue}{$\blacktriangleright$} Second local maximum}}};
            \node[rotate=0, line width=0.05mm, draw=white, anchor=west] at (3.5,-2.3) {{\footnotesize{\textcolor{red}{$\blacktriangledown$} Third local maximum}}};
            \node[rotate=0, line width=0.05mm, draw=white, anchor=west] at (3.5,-2.8) {{\footnotesize{\textcolor{red}{$\circ$} Best local maximum}}};
            \node[rotate=0, line width=0.05mm, draw=white] at (5,-0.4) {{\footnotesize{Newton method}}};
    		 \draw [-stealth](1,-2.9) -- (0.2,-2.65);
            \draw [-stealth](1,-2.9) -- (1.,-2.6);
            \draw [-stealth](1,-2.9) -- (1.55,-2.6);
    		\node[rotate=0, line width=0.05mm, draw=white] at (1,-3.1) {{\small{Exhaustive search}}};
            
            \node[rotate=0, line width=0.05mm, draw=white, anchor=west] at (-5.75,3) {{\footnotesize{Bandwidth}}};
            \node[rotate=0, line width=0.05mm, draw=white, anchor=west] at (1.5,3) {{\footnotesize{Correlation landscape}}};
		\end{tikzpicture} 
    
    \caption{Evolution of a 1D correlation landscape under progressive bandwidth expansion. Each curve corresponds to a different bandwidth, with lower bandwidths producing smoother objectives and typically fewer local maxima. A coarse exhaustive search (evaluated on a discretization depicted in green) at low bandwidth identifies local maxima. As the bandwidth increases, these local maxima are tracked across scales using frequency marching, while local optimization refines the solution at each stage. This strategy avoids exhaustive search on highly oscillatory high-frequency landscapes while achieving high-precision alignment. Left insets show the corresponding low-pass-filtered particle, illustrating how low angular bandwidth suppresses high-frequency detail.  }\label{fig:illustration}
\end{figure}
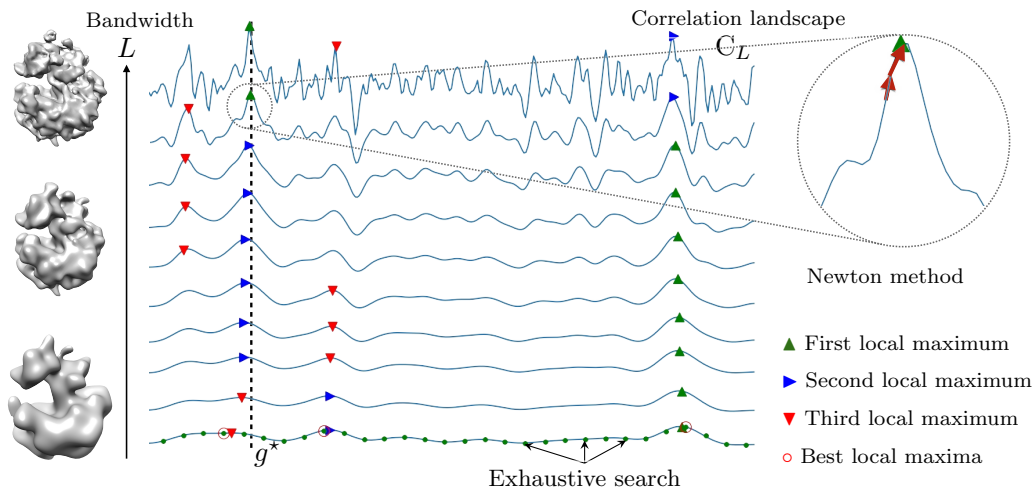

\subsection{Related work}

Estimating three-dimensional rotations and translations is a central problem in cryo-EM and cryo-ET alignment. 
In practical refinement pipelines, pose estimation is usually performed through a combination of coarse search, local refinement, and progressive resolution increase. 
RELION~\cite{scheres2012relion} has played a central role in this ecosystem, both for single-particle refinement and, more recently, for end-to-end cryo-ET subtomogram averaging in RELION-5~\cite{burt2024image}. 
In RELION-style workflows, subtomogram refinement is embedded in a larger processing pipeline that includes particle extraction, half-set refinement, CTF refinement, polishing or tomography refinement, classification, reconstruction, and post-processing. 
\method accelerates the main computational bottlenecks in this pipeline: rotational pose refinement inside a subtomogram-averaging workflow.

Fast algorithms for pose search often leverage harmonic analysis on transformation spaces. For translations, correlations can be evaluated efficiently by the Euclidean FFT. 
For rotations, FFTs on \(\SO(3)\) based on Wigner--\(D\) expansions enable fast evaluation of correlations over dense rotational grids~\cite{kostelec_ffts_2008}. 
Related spherical and rotational harmonic transform libraries, including S2FFT~\cite{mcewen2015novel,price_differentiable_2024}, provide accelerated transforms on the sphere and on \(\SO(3)\). 
More generally, fast Fourier-transform-based approaches have been used for template matching and alignment over transformation spaces~\cite{kovacs2002fast,kovacs2003fast,estrozi2008fast,bartesaghi_classification_2008}. Still, dense high-resolution searches over rotations and translations remain expensive when repeated for every particle across multiple STA refinement iterations.
Our implementation builds on the ball-harmonic expansion routines of Kileel et al.~\cite{kileel2025fast}, extending their framework with a GPU-integrated PyTorch pipeline and accelerated NUFFT-based projections via cuFINUFFT~\cite{shih2021cufinufft}.

The use of low-resolution information to guide high-resolution refinement is well established in cryo-EM. FREALIGN performs local continuous refinement of pose parameters against resolution-truncated references, exploiting the fact that low-resolution signal is often sufficient to guide pose refinement while the final estimate is not tied to a coarse angular grid~\cite{grigorieff_frealign_2007}. 
cryoSPARC later introduced a branch-and-bound refinement strategy for global pose search, based on the principle that poses scoring poorly at low resolution are unlikely to become optimal at high resolution~\cite{punjani2017cryosparc}. A similar frequency-marching idea has been used in ab initio single-particle reconstruction by Barnett et al.~\cite{barnett2017rapid}.

This intuition is closely related to our use of a low-angular-bandwidth objective to identify candidate basins. 
What we do is different: after a single low-bandwidth exhaustive search, we refine candidates continuously on \(\SO(3)\) using closed-form Newton steps applied to the Wigner--\(D\) expansion, and we provide deterministic conditions under which the tracked basins persist across the bandwidth schedule. In this way we minimize the coupling between target resolution and bandwidth.

There also exist methods specifically designed for subtomogram alignment in cryo-ET. Some approaches accelerate alignment by combining FFT-based searches over rotations and translations~\cite{kovacs2003fast,estrozi2008fast,bartesaghi_classification_2008}, while others incorporate features of the tomographic imaging process, such as the missing wedge~\cite{shatsky2013optimal} or spatially varying density variance~\cite{xu2012high}. 
Chen et al.~\cite{chen_fast_2013} introduce a fast subtomogram alignment method that is robust to the choice of initial reference; in particular, the initial reference can be formed from the data itself and is iteratively updated by averaging the aligned subtomograms. 
These works address the broader subtomogram-alignment problem, including reference generation, missing-information effects, or translation search. 
Continuous pose refinement has also been explored in cryo-EM. 
Zehni et al.~\cite{zehni2020joint} use alternating gradient-based optimization and ADMM for joint map--orientation estimation. 

\section{Method} \label{sec:methods}

We start by introducing the alignment problem and its spectral formulation; we then present our coarse-to-fine optimization strategy. Motivated by subtomogram alignment in cryo-electron tomography, the main text develops the rotational component of the method where $\Gc=\SO(3)$. The full six-dimensional alignment problem is handled in practice by alternating this rotational update with FFT-based translation updates, described in Appendix~\ref{sec:appendix-6d-search}.
The rotation $g\in\Gc$ acts on a function $f:\mathbb{R}^3\to \mathbb{R}$ as
\begin{equation*}
    (g \circ f)(x) = f\left(g^{-1}x\right).
\end{equation*}
We let $\Bc_3 \subset \mathbb{R}^3$ denote the unit ball and as before $f \in L^2(\Bc_3)$ the observed subtomogram and $h \in L^2(\Bc_3)$ the reference template. The $L^2(\Bc_3)$ inner product is $\langle f,h\rangle \eqdef \int_{\Bc_3} f(x)\,\overline{h(x)}\,dx$.

The orientation estimation problem is formulated as
\begin{equation}
\label{eq:maxcc_rot}    
\hat g \in \argmax_{g\in \Gc} \CC(g), 
\qquad 
\CC(g) \eqdef \left\langle f ,g \circ h\right\rangle.
\end{equation}

The analysis in Section~\ref{sec:analysis} is carried out intrinsically on \(\SO(3)\), independently of any choice of coordinates. The explicit Euler-angle chart used in the implementation, together with the resulting closed-form derivatives, is described in Section~\ref{sec:continuous-opt} and Appendices~\ref{app:parametrization}--\ref{app:d-matrix}.

\subsection{Ball-harmonic decomposition of the rotational cross-correlation}

We decompose $f$ and $h$ in a ball-harmonic basis, which separates radial and angular variables. Under a rotation, the radial basis functions remain unchanged, while for each degree $\ell$, the spherical-harmonic coefficients are mixed by the Wigner--\(D\) matrices. After summing over the radial index, the rotational cross-correlation can therefore be expressed as a Fourier series on \(\SO(3)\) in the Wigner--\(D\) basis.

Truncating the angular content at degree $L$, we approximate the rotational cross-cor\-re\-la\-tion~\eqref{eq:maxcc_rot} by
\begin{equation}\label{eq:cc-L}
\CC_{L}(g) \eqdef  
\sum_{\ell = 0}^{ L} 
\sum_{|m|\leq \ell} 
\sum_{|m'|\leq \ell} 
\sigma_{\ell,m,m'} \, D_{m m'}^\ell(g),
\quad g \in \Gc.
\end{equation}
Here $D^\ell_{m m'}$ are the matrix coefficients of the irreducible unitary representation of $\SO(3)$ of degree $\ell$. By the Peter--Weyl theorem, they form the analogue of the Fourier basis on $\SO(3)$. The expansion in \eqref{eq:cc-L} is therefore a Fourier series for the rotational cross-correlation, and truncation at degree $L$ amounts to retaining only frequencies up to $L$. In particular, $\CC_L \to \CC$ in $L^2(\SO(3))$ as $L \to \infty$.
The coefficients \(\sigma_{\ell,m,m'}\) are obtained by summing products of the
ball-harmonic coefficients of \(f\) and \(h\) over the retained radial index;
their explicit form is given in Appendix~\ref{sec:appendix-rotation-cc}.

The cutoff \(L\) is an angular degree cutoff in the Wigner--\(D\) expansion, not
a spatial frequency in inverse pixels or inverse \angstrom. It controls the angular
complexity of the rotational score. Its relation to conventional spatial
resolution depends on object radius and is therefore not described by a single
global \angstrom{} value.
 
The representation in Equation~\ref{eq:cc-L} is also convenient for optimization: first- and second-order derivatives can be evaluated by summing \(\Oc(L^3)\) Wigner terms. In the coarse-to-fine scheme, we start from a small cutoff $L_0$---yielding a smoothed, low-fre\-quen\-cy objective---and progressively increase $L$, to enrich the angular content of $\CC_L(g)$.
The approximation at each new bandwidth is obtained simply by extending the sum over the precomputed coefficients $\sigma_{\ell,m,m'}$ to include higher degrees; no recomputation of lower-degree coefficients is needed.
 
\subsection{Continuous optimization}\label{sec:continuous-opt}
Because the cross-correlation $\CC$ is non-convex with multiple critical points,  standard local optimization methods will in general not converge to its global maxima from an arbitrary initialization. This, however, becomes possible provided that a good initial estimate is available. Finding a good initial estimate is in turn easier for the bandlimited cross-correlation $\CC_L$, since lower-bandwidth functions generically have fewer critical points and a smoother landscape, which translates into a larger basin of attraction of a global maximum. 

A standard modern approach to continuous optimization is to use automatic differentiation and gradient ascent: it only requires computing gradients and it typically makes progress from a coarse initialization. However, on landscapes such as those induced by bandlimited trigonometric polynomials, second-order methods can be orders of magnitude faster than gradient ascent, which may require many small steps to leverage the curvature information. This is especially relevant in our setting, where the cost of evaluating the objective and its derivatives is dominated by the same $\Oc(L^3)$ Wigner--\(D\) terms, which makes the asymptotic per-iteration scaling in $L$ the same for first- and second-order methods. We obtain these derivatives in closed form; automatic differentiation would instead require constructing and backpropagating through the computational graph of the recursive Wigner--\(D\) construction. Moreover, because $\Gc$ has a fixed low dimension (three for $\SO(3)$), the Newton step only requires inverting a $3\times 3$ matrix, which can be done in explicit form. 

Iteration count is therefore the main factor determining the difference in runtime.  Newton's method typically converges in only a few iterations once initialized in the desired basin of attraction; see Figure~\ref{fig:newton_v_gd}.

\paragraph{Coordinates on $\SO(3)$} While $\SO(3)$ admits an intrinsic Riemannian geometry, which would allow one to compute a Riemannian gradient and Hessian together with the corresponding retraction or exponential-map update, we found that working in a local Euclidean coordinate chart works well in practice. In our implementation we optimize directly in wrapped ZYZ Euler-angles $\theta = (\alpha, \beta, \gamma)$, mapping to rotations via $g(\theta) = r_z(\alpha)\,r_y(\beta)\,r_z(\gamma)$ (Equation~\ref{eq:Euler}). We use the ZYZ convention because the Wigner--\(D\) matrices factor as $D^\ell_{m m'}(\alpha,\beta,\gamma) = e^{-im\alpha}\,d^\ell_{m m'}(\beta)\,e^{-im'\gamma}$, so that the $\alpha$ and $\gamma$ derivatives act diagonally and the gradient and Hessian of $\CC_L$ are available in closed form. This chart has coordinate singularities only at $\beta\in\{0,\pi\}$; since the continuous optimization is initialized at a rotation near a maximizer, the iterates remain in a neighborhood where the parameterization is well-conditioned. 
In this chart, the update takes the usual form
\begin{equation*}
    \theta^{t+1} = \theta^t - H_{\CC_L}^{-1}(\theta^t)\,\nabla \CC_L(\theta^t),
    \qquad g^t = g(\theta^t),
\end{equation*}
where $H_{\CC_L}(\theta)\eqdef\nabla^2 \CC_L(\theta)$ denotes the Hessian.
\subsection{Initialization via low-resolution SOFFT}\label{sec:sofft}

Identifying the global maximum of $\CC_L$ using continuous  optimization depends on suitable initialization. For this we perform a grid search on the low-bandwidth correlation $\CC_{L_0}$, which is cheap to evaluate.

We use an oversampled equiangular grid with approximately $(KL_0)^3$ rotations, where $K\geq 1$ is the oversampling factor. The values of $\CC_{L_0}$ on this grid are computed using an inverse $\SO(3)$ Fourier transform (SOFFT) \cite{kostelec_ffts_2008}. On the natural equiangular grid ($K=1$), SOFFT requires $\Oc(L_0^4)$ operations; 
for $K > 1$ optimized implementations achieve $\Oc((KL_0)^3 +L_0^4)$ \cite{potts_fast_2009}.
We retain the $P$ strongest local maxima as candidate initializations.

\subsection{Coarse-to-fine optimization algorithm}

The complete alignment procedure combines a single exhaustive search at $L_0$ with frequency-marched continuous optimization. The final estimate is reported at angular cutoff $L_J$, chosen such that $\CC_{L_J}$ approximates the full correlation to the desired accuracy. Algorithm~\ref{algo:coarse-to-fine} summarizes the procedure.

In practice, we run at most $M_{\mathrm{iter}}$ Newton steps per band and stop early when the gradient norm, the parameter increment, or the relative change in the objective fall below a prescribed tolerance. When applying \method{} to subtomogram alignment~(\ref{sec:exp_STA}), we use a fixed bandwidth schedule starting from $L_0 = 30$, for which the exhaustive SOFFT search remains cheap. Intermediate bandwidths are chosen to keep the per-step perturbation within the stability bounds of our convergence analysis (Section~\ref{sec:benefit_marching}). We retain $P = 10$ candidates with oversampling factor $K = 2$ for the initial SOFFT search.

\begin{algorithm}
        \caption{\method{}, a Multi-band Angular Template Matching Algorithm}\label{algo:coarse-to-fine}
        \begin{algorithmic}[1]
            \Require
                \begin{tabular}[t]{@{}l@{}}
                    Angular cutoffs $\{L_0, L_1, \dots, L_J\}$ \\
                    Number of Newton iterations $M_{\mathrm{iter}}$ \\
                    Number of candidates $P$ \\
                    Oversampling factor $K$
                \end{tabular}
            \State Perform exhaustive search with SOFFT at angular cutoff $L_0$ with oversampling factor $K$
            \State Retain the $P$ strongest local maxima $\{g_p^{(0)}\}_{p=1}^{P}$
            \State Optionally refine each retained candidate on \(\CC_{L_0}\), and redefine the refined outputs as \(\{g_p^{(0)}\}_{p=1}^P\).
           \For {$j$ in $\{1,\hdots, J\}$}
              \For {$p$ in $\{1,\hdots, P\}$}
                \State $g_p^{(j)} \leftarrow$ result of $M_{\mathrm{iter}}$ Newton steps on $\CC_{L_j}$ initialized at $g_p^{(j-1)}$
              \EndFor
            \EndFor
\State \textbf{return} $\displaystyle \argmax_{p\in\{1,\dots,P\}} \CC_{L_J}\!\left(g_p^{(J)}\right)$
        \end{algorithmic}
\end{algorithm}

\section{Theoretical analysis}
\label{sec:analysis}

The practical efficiency of \method{} (Algorithm~\ref{algo:coarse-to-fine})
comes from three features: the search space is only three-dimensional,
the derivatives of the correlation objective are available in closed form,
and local second-order refinement is cheap.

In this section we prove that under realistic assumptions the frequency-marching scheme defined in Algorithm~\ref{algo:coarse-to-fine} returns a final candidate that is arbitrarily close to the global maximizer. The proof formalizes the idea that selected strongly concave coarse-level maximizer can be continued across successive bandwidth levels, with controlled drift, provided the inter-level perturbations are small. This requires handling two failure modes. One is that the maximizer inside a tracked basin stops being unique. The other is that a point outside the tracked region \(S\) (see Section~\ref{sec:theory-setting} for the definition) becomes globally optimal at the final bandwidth. The first is controlled by strong concavity plus small $C^2$ perturbations; the second is ruled out by a set-gap condition and a global sup-norm bound.

\subsection{Setting} \label{sec:theory-setting}
We equip $\SO(3)$ with its bi-invariant metric, unique up to scaling
\cite{MILNOR1976293}. All distances, geodesic balls, gradients $\nabla F$, Hessians $\nabla^2 F$, and operator norms are taken with respect to this metric.
Let $S\subset \SO(3)$ denote a compact tracked region. Intuitively, \(S\) represents the union of basins around the coarse candidates retained after the low-bandwidth search.

We select a bandwidth schedule $L_0 < L_1 < \cdots < L_J$ and write
$$
F_j \eqdef \CC_{L_j} \ \ (j=0,1,\ldots,J),
\qquad
E_j \eqdef F_j - F_{j-1} \ \ (j=1,\ldots,J).
$$
The role of \(E_j\) is to measure the perturbation introduced when passing from bandwidth \(L_{j-1}\) to bandwidth \(L_j\). To exclude the possibility that a point outside the tracked region \(S\) becomes globally optimal at the final bandwidth, we impose a coarse-level set-gap condition and require the accumulated perturbation from \(F_0\) to \(F_J\) to be smaller than this gap.

\begin{definition}
Let $F : \SO(3) \to \R$ be continuous and $S \subset \SO(3)$ be compact. The set gap of $F$ with respect to $S$ is
\begin{align}     \label{eq:setgap_def}
    \Gamma_F(S) \eqdef \max_{g\in S}F(g) - \sup_{g\notin S} F(g).
\end{align}
\end{definition}
Thus $\Gamma_F(S)$ measures the margin by which the best value inside $S$ exceeds the best value outside $S$.
Note that if $\Gamma_F(S) > 0$, then $S$ contains every global maximizer of $F$. Indeed, any $g \notin S$ satisfies $F(g) \leq \sup_{g'\notin S} F(g') < \max_{g'\in S} F(g')$, so $g$ cannot be a global maximizer.

\subsection{Excluding failure modes}

The first lemma shows that a positive set gap remains positive under sufficiently small perturbations.

\begin{lemma}[Set-gap stability]
\label{lem:setgap_stability}
Let $F, \widetilde F:\SO(3) \to \R$ be continuous and let $S\subset \SO(3)$ be compact.
Assume the set gap $\Gamma_F(S)$ defined in \eqref{eq:setgap_def} satisfies $\Gamma_F(S)>0$. If
\begin{equation}
\big\| F - \widetilde F \big\|_\infty \le \varepsilon
\qquad\text{and} \qquad
2\varepsilon < \Gamma_F(S),
\label{eq:setgap_eps_cond}
\end{equation}
then every global maximizer $\tilde g^\star\in\argmax \widetilde F$ lies in $S$.
\end{lemma}

\begin{proof}
   Since $\SO(3)$ is compact and $\widetilde{F}$ is continuous, $\widetilde{F}$ attains its maximum. 
    Define $a_S := \max_{g \in S} F(g)$ and $a_{S^c} := \sup_{g \notin S} F(g)$, so $a_S - a_{S^c} = \Gamma_F(S) > 0$ by assumption. Since $S$ is compact and $F$ is continuous, there exists $g_S \in S$ such that $F(g_S) = a_S$. For any $g \notin S$,
    \begin{align}
        \widetilde{F}(g) 
        \leq F(g) + \varepsilon 
        \leq a_{S^c} + \varepsilon
        \leq a_S - \Gamma_F(S) + \varepsilon,
    \end{align}
    and also
    $$
        \widetilde{F}(g_S) \geq F(g_S) - \varepsilon = a_S - \varepsilon.
    $$
    Subtracting, 
    $$
       \widetilde{F}(g_S) - \widetilde{F}(g) \geq (a_S - \varepsilon) - (a_S - \Gamma_F(S) + \varepsilon) = \Gamma_F(S) - 2\varepsilon > 0.
    $$
    Thus $\widetilde{F}(g_S) > \widetilde{F}(g)$ for $g \notin S$ and every global maximizer of \(\widetilde F\) lies in \(S\).
\end{proof}
In our application, this lemma is a way to exploit the spectral decay with frequency. Concretely, if the accumulated high-frequency tail satisfies
\[
\|\CC_{L_J}-\CC_{L_0}\|_\infty < \Gamma_{\CC_{L_0}}(S)/2,
\]
then every global maximizer of \(\CC_{L_J}\) lies in \(S\). In particular, a point outside \(S\) cannot become globally optimal at the final bandwidth.

To control the tracked maxima inside $S$, we assume upper and lower bounds on the Hessian. 

\begin{definition}[Second-order concavity and curvature bounds] Let $U \subset \SO(3)$ be compact and geodesically convex, let $\Omega \subset \SO(3)$ be an open neighborhood of $U$, let $F\in C^2(\Omega)$ and let $\mu, M > 0$. We say $F$ is $\mu$-strongly concave on $U$ if
$$
\inprod{\nabla^2 F(g)[X]}{X}\le -\mu
$$ 
for every $g \in U$ and every unit tangent vector $X \in T_g\SO(3)$.

We say that the Hessian of $F$ is $M$-smooth on $U$ if 
$$
\inprod{\nabla^2 F(g)[X]}{X}\ge -M
$$
for every $g \in U$ and every unit tangent vector $X \in T_g\SO(3)$.
\end{definition}

Since $\SO(3)$ has positive convexity radius, every sufficiently small closed geodesic ball is geodesically convex. Together with Hessian bounds this allows us to control the objective around critical points. 

\begin{lemma}[Two-sided quadratic bounds around a critical point]
\label{lem:2sided}
Let $U$ be geodesically convex, let $\Omega\subset \SO(3)$ be an open neighborhood of $U$, and let $F\in C^2(\Omega)$, and suppose $g_0\in U$ satisfies $\nabla F(g_0)=0$.
Assume $F$ is $\mu$-strongly concave and its Hessian is $M$-smooth on $U$. Then for every $g\in U$,
$$
F(g_0) - \frac{M}{2} d(g,g_0)^2 \le F(g)\le F(g_0)-\frac{\mu}{2}\,d(g,g_0)^2.
$$
\end{lemma}
\begin{proof}
Fix $g\in U$. By geodesic convexity of \(U\), there exists a unit-speed geodesic \(\gamma:[0,\ell]\to U\) from \(g_0\) to \(g\), where
\(\ell=d(g,g_0)\). \\
Let $h(t) \eqdef F(\gamma(t))$. Then $h'(0)=\inprod{\nabla F(g_0)}{\dot\gamma(0)}=0$ and
$$
h''(t)=\inprod{\nabla^2 F(\gamma(t))[\dot\gamma(t)]}{\dot\gamma(t)}.
$$
Since $\|\dot\gamma(t)\|=1$ and $F$ is $\mu$-strongly concave and $M$-smooth on $U$, we have for all $t\in[0,\ell]$,
$$
-M\le h''(t)\le -\mu.
$$
Integrate $h''(t)\ge -M$ from $0$ to $t$ to get $h'(t)\ge -Mt$, and integrate again from $0$ to $\ell$ to get
$$
h(\ell)-h(0)=\int_0^\ell h'(t)\,dt \ge \int_0^\ell (-Mt)\,dt = -\frac{M}{2}\ell^2.
$$
Thus $h(\ell)\ge h(0)-\frac{M}{2}\ell^2$, i.e., \ $F(g)\ge F(g_0)-\frac{M}{2}d(g,g_0)^2$.
The upper bound can be derived analogously using $h''(t)\le -\mu$.
Indeed, \(h''(t)\le -\mu\) gives \(h'(t)\le -\mu t\), and hence
\[
h(\ell)-h(0)\le -\frac{\mu}{2}\ell^2.
\]
\end{proof}

Next we show that a strongly concave basin persists under small $C^2$ perturbations, in the sense that the maximizer remains unique, stays in the interior, and moves by a controlled amount.

\begin{lemma}[Basin persistence under a $C^2$ perturbation]
\label{lem:basin_persist}
Let \(r>0\) be below the convexity radius of \(\SO(3)\), and let \(U=\overline{B(g^\star,r)}\) be a closed geodesically convex ball. Let $\Omega \subset \SO(3)$ be an open neighborhood of $U$,
and let $F, E \in C^2(\Omega)$.
Assume \(F\) is \(\mu\)-strongly concave on \(U\) and that \(g^\star\) maximizes \(F\) on \(U\).
Let $\widetilde{F} \eqdef F + E$ and set
$$
    a \eqdef \norm{\nabla E}_{L^\infty(U)}, \qquad b \eqdef \sup_{g\in U}\|\nabla^2 E(g)\|_{\mathrm{op}}.
$$
If $b<\mu$ and $a<(\mu-b)\,r/2$, then $\widetilde F$ is $(\mu-b)$-strongly concave on $U$ with a \emph{unique} maximizer $\tilde g^\star \in \mathrm{int}(U)$ such that
$
d(\tilde g^\star,g^\star) \le {a}/{(\mu-b)}.
$
\end{lemma}
\begin{proof}
Since \(g^\star\) is the center of the ball \(U\), it lies in
\(\operatorname{int}(U)\). Because it maximizes \(F\) on \(U\) and
\(F\in C^2(\Omega)\), we have \(\nabla F(g^\star)=0\).

To verify strong concavity, note that for any unit tangent $X$ at any $g \in U$,
$$
    \inprod{\nabla^2 \widetilde F(g)[X]}{X}
    = \inprod{\nabla^2 F(g)[X]}{X}+\inprod{\nabla^2 E(g)[X]}{X}
    \leq -\mu + b = -(\mu-b).
$$
By continuity and compactness, $\widetilde{F}$ attains a maximum on $U$.
Because $\widetilde{F}$ is $(\mu-b)$-strongly concave on the geodesically convex $U$,
it can have at most one maximizer in $U$ (if there were two, strict concavity along the
connecting geodesic would be contradicted).

We now show that the maximizer cannot be on the boundary $\partial U$.
Let \(y\in\partial U\). Since \(U=\overline{B(g^\star,r)}\) is a closed geodesic ball with radius below the convexity radius, there is a unit-speed minimizing geodesic \(\gamma:[0,r]\to U\) from \(g^\star\) to \(y\).
Define $h(t) = \widetilde F(\gamma(t))$.
Then $h'(0) = \inprod{\nabla \widetilde F(g^\star)}{\dot\gamma(0)} = \inprod{\nabla E(g^\star)}{\dot\gamma(0)}$ so $|h'(0)|\le a$.
We already showed $h''(t) \leq -(\mu - b)$.
Integrating $h''(t)\le -(\mu-b)$ twice gives
$$
h(r) \leq h(0) + r\,h'(0) - \frac{\mu - b}{2}r^2 \le h(0) + ar - \frac{\mu - b}{2}r^2.
$$
Since $a < (\mu - b)r/2$, we have $a r -\frac{\mu - b}{2}r^2 < 0$, hence $h(r) < h(0)$. Thus no boundary point can maximize $\widetilde{F}$, so the maximizer lies in $\mathrm{int}(U)$.

Finally, to bound the displacement, let $\tilde g^\star$ be the unique maximizer of $\widetilde F$ in $U$ and $\gamma:[0,\ell] \to U$ the unit-speed geodesic from $\gamma(0) = g^\star$ to $\gamma(\ell) = \tilde g^\star$,
so $\ell=d(g^\star,\tilde g^\star)$. Since \(\tilde g^\star\in \operatorname{int}(U)\), the first-order
optimality condition gives \(\nabla \widetilde F(\tilde g^\star)=0\).
Define $\psi(t) = \inprod{\nabla \widetilde F(\gamma(t))}{\dot\gamma(t)}$.
Then $\psi(\ell) = 0$ because $\nabla\widetilde F(\tilde g^\star) = 0$, and
$$
\psi'(t)=\inprod{\nabla^2 \widetilde F(\gamma(t))[\dot\gamma(t)]}{\dot\gamma(t)}\le-(\mu-b).
$$
Integrating gives $0 = \psi(\ell) \leq \psi(0) - (\mu - b)\ell$, hence $(\mu - b)\ell \leq \psi(0)$.
Since $\nabla \widetilde F(g^\star) = \nabla E(g^\star)$, we get $\psi(0) = \inprod{\nabla E(g^\star)}{\dot\gamma(0)} \leq \|\nabla E(g^\star)\|\le a$.
Therefore $\ell \leq a/(\mu - b)$.
\end{proof}

\subsection{Main result}

Assumptions~\ref{ass:A-initialization}--\ref{ass:C-setgap} concern the family of correlation landscapes
$\{F_j\}_{j=0}^J$.
Additionally, Assumption~\ref{ass:D-new_accuracy} abstracts away the local convergence analysis of the particular Newton implementation. 

The proof shows that Assumptions~\ref{ass:A-initialization}--\ref{ass:C-setgap} imply persistence of the tracked
basins and containment of all global maximizer of $F_J$ inside $S$.
Together with Assumption~\ref{ass:D-new_accuracy} and the
two-sided second-order bound of Lemma~\ref{lem:2sided} this implies near-optimality of the best final candidate returned by Algorithm~\ref{algo:coarse-to-fine}. 

For each $j=1,\dots,J$ and each basin index $p$, Algorithm~\ref{algo:coarse-to-fine} applies \(M_{\mathrm{iter}}\) Newton steps
to \(F_j\), initialized from \(g_p^{(j-1)}\).
Let $g_p^{(j)}$ denote the returned iterate.

\begin{assumption}
\begin{enumerate}[label=(\Alph*),ref=(\Alph*),series=assumptions]
\item (Initialization)\label{ass:A-initialization} Fix $r>0$ below the convexity
radius of $\SO(3)$, so that every closed geodesic ball of radius at most $r$ is
geodesically convex.
The initialization produces one point $g_p^{(0)}$ for each $p=1,\dots,P$.
Around each such point we define the tracked basin
\[
    U_p \eqdef \overline{B(g_p^{(0)},r)} \subset \SO(3),
\]
with $U_p$ being pairwise disjoint, and $g_p^{(0)}$ the unique maximizer of
$F_0$ on $U_p$.
We let $S \eqdef \bigcup_{p = 1}^P U_p$. Since $S$ is a finite union of compact
sets, it is compact.

\item (Small $C^2$ inter-level perturbations.)\label{ass:B-perturb}
Assume that \(F_j\) is continuous on \(\SO(3)\) for every \(j=0,\dots,J\), and \(C^2\) on \(S\). Let $E_j \eqdef F_j - F_{j-1}$ and
\[
    a_j \eqdef \norm{\nabla E_j}_{L^\infty(S)},\qquad
    b_j \eqdef \sup_{g\in S}\|\nabla^2 E_j(g)\|_{\mathrm{op}}, \qquad
    \delta_j \eqdef \norm{E_j}_{L^\infty(\SO(3))}.
\]
We have
\begin{equation}
\sum_{j=1}^J b_j \le \frac{\mu}{2},
\qquad
a_j < \frac{\mu r}{8},
\qquad
\sum_{j=1}^J \frac{2a_j}{\mu} \le \frac{r}{2},
\label{eq:step_smallness}
\end{equation}
and $F_0$ is $\mu$-strongly concave on each $U_p$.

\item (Set-gap condition.)\label{ass:C-setgap}
We have $\Gamma_{F_0}(S)>0$ and
\begin{equation}
2\sum_{j=1}^J \delta_j < \Gamma_{F_0}(S).
\label{eq:tail_gap_cond}
\end{equation}
In particular,
\[
\|F_J-F_0\|_\infty \le \sum_{j=1}^J \delta_j < \frac{\Gamma_{F_0}(S)}{2}.
\]

\item (Finite Newton accuracy.)\label{ass:D-new_accuracy} 
There is accuracy $\tau \in [0, r/2]$ such that, for every basin
index $p=1,\dots,P$ and every level $j=1,\dots,J$, the iterate $g_p^{(j)}$
returned by the $M_{\mathrm{iter}}$ Newton steps on $F_j$ (initialized at
$g_p^{(j-1)}$) satisfies
\begin{equation}
  d\bigl(g_p^{(j)},\, m_{p,j}\bigr) \le \tau,
  \label{eq:newton_tau}
\end{equation}
where $m_{p,j}\in\operatorname{int}(U_p)$ denotes the maximizer of $F_j$ on
$U_p$.
\end{enumerate}
\end{assumption}
Note that by Assumption~\ref{ass:A-initialization}, we assume that the relevant coarse basins have already been selected.
\begin{proposition}
\label{prop:AC_consequences}
Under Assumptions~\ref{ass:A-initialization}--\ref{ass:C-setgap}, for every basin index $p$ and level $j=1,\dots,J$ we have:
\begin{enumerate}[label=(\roman*)]
\item $F_j$ is $\left(\mu-\sum_{k=1}^j b_k\right)$-strongly concave on $U_p$, and hence in particular $(\mu/2)$-strongly concave on $U_p$.
\item $F_j$ has a unique maximizer $m_{p,j}\in\mathrm{int}(U_p)$.
\item The basin maximizer satisfy the drift bounds
\[
d(m_{p,j},m_{p,j-1})
\le \frac{a_j}{\mu-\sum_{k=1}^j b_k}
\le \frac{2a_j}{\mu},
\]
and therefore
\[
d(m_{p,j},m_{p,0})
\le \sum_{k=1}^j \frac{2a_k}{\mu}
\le \frac r2.
\]
\item Every global maximizer of $F_J$ lies in $S$.

\end{enumerate}
\end{proposition}

\begin{proof}
  \textbf{Part (i):}
  Throughout, we write $m_{p,0} = g_p^{(0)}$ for the initial maximizer of $F_0$ on $U_p$.
Fix $p$.
On $U_p$ we have
\[
\inprod{\nabla^2 F_j(g)[X]}{X}
=
\inprod{\nabla^2 F_0(g)[X]}{X}
+\sum_{k=1}^j \inprod{\nabla^2 E_k(g)[X]}{X}
\le -\mu + \sum_{k=1}^j b_k,
\]
 so $F_j$ is $\left(\mu-\sum_{k=1}^j b_k\right)$-strongly concave on $U_p$.
Since $\sum_{k=1}^J b_k \le \mu/2$, this implies (i).

\textbf{Part (ii)--(iii):}
We prove (ii)--(iii) by induction on $j$.
The case $j=0$ is exactly Assumption~(A).
Assume the claim holds at level $j-1$.
Set
\[
\mu_{j-1}\eqdef \mu-\sum_{k=1}^{j-1}b_k \ge \frac{\mu}{2}.
\]
By the induction hypothesis,
\[
d(m_{p,j-1},m_{p,0}) \le \sum_{k=1}^{j-1}\frac{2a_k}{\mu}\le \frac r2.
\]
Hence the smaller ball
\[
U' \eqdef \overline{B(m_{p,j-1},r/2)}
\]
is geodesically convex and contained in $U_p$. Indeed, if \(g\in U'\), then
\[
d(g,m_{p,0})\le d(g,m_{p,j-1})+d(m_{p,j-1},m_{p,0})
\le r/2+r/2=r.
\]
We shrink to $U'$ so that the perturbation lemma can be applied around the previously tracked maximizer while staying inside the original basin $U_p$. 
Since \(U'\subset U_p\) and \(m_{p,j-1}\) is the unique maximizer of \(F_{j-1}\) on \(U_p\), it is also the unique maximizer of \(F_{j-1}\) on \(U'\).
Since
\[
\mu_{j-1}-b_j \ge \mu-\sum_{k=1}^J b_k \ge \frac{\mu}{2},
\qquad
a_j < \frac{\mu r}{8} \le (\mu_{j-1}-b_j)\frac r4,
\]
we can apply Lemma~\ref{lem:basin_persist} to $F_{j-1}$ perturbed by $E_j$ on $U'$, with strong concavity constant $\mu_{j-1}$ and radius $r/2$. 
Lemma~\ref{lem:basin_persist} shows that
$m_{p,j}\in \mathrm{int}(U')$ is a maximizer of $F_j$ on $U'$ satisfying
\[
d(m_{p,j},m_{p,j-1}) \le \frac{a_j}{\mu_{j-1}-b_j}
= \frac{a_j}{\mu-\sum_{k=1}^j b_k}
\le \frac{2a_j}{\mu}.
\]
The maximizer produced by Lemma~\ref{lem:basin_persist} lies in
\(\operatorname{int}(U')\), hence it is a critical point of \(F_j\).
Since \(F_j\) is strongly concave on the geodesically convex set \(U_p\),
any critical point in \(U_p\) is the unique global maximizer of \(F_j\) on
\(U_p\). Therefore this point is the unique maximizer \(m_{p,j}\) of \(F_j\)
on \(U_p\).
Summing in $j$ yields the cumulative drift bound in (iii). This proves (ii)--(iii); in words, each tracked basin persists across levels, its maximizer remains unique and interior, and the cumulative drift stays controlled.

\textbf{Part (iv):}
Finally,
\[
\|F_J-F_0\|_\infty
\le \sum_{j=1}^J \delta_j
< \frac{\Gamma_{F_0}(S)}{2}
\]
by Assumption~(C), so Lemma~\ref{lem:setgap_stability} implies that every global maximizer of $F_J$ lies in $S$.
This proves (iv).
\end{proof}

\begin{theorem}[Main result: near-optimality in objective value]
\label{thm:multiband_correct}
Assume~\ref{ass:A-initialization}--\ref{ass:D-new_accuracy} and let
\[
\hat p \in \argmax_{p=1,\dots,P} F_J\bigl(g_p^{(J)}\bigr) \text{ and } \hat g \eqdef g_{\hat p}^{(J)}
\]
be the output of \method{} (Algorithm~\ref{algo:coarse-to-fine}).
Define
\[
M_J \eqdef \sup_{g\in S}\|\nabla^2F_J(g)\|_{\mathrm{op}} < \infty.
\]
Then
\[
0 \le \max_{g\in\SO(3)}F_J(g)-F_J(\hat g)\le \frac{M_J}{2}\tau^2,
\]
where $\tau$ is the uniform Newton accuracy from Assumption~\ref{ass:D-new_accuracy}.
\end{theorem}

\begin{proof}
By Proposition~\ref{prop:AC_consequences}, for each $p$ the function $F_J$ has a unique maximizer
$m_{p,J}\in \mathrm{int}(U_p)$, $F_J$ is $(\mu/2)$-strongly concave on $U_p$,
and every global maximizer of $F_J$ lies in $S$.

By Assumption~\ref{ass:D-new_accuracy}, for every $p$,
\[
d\bigl(g_p^{(J)},m_{p,J}\bigr)\le \tau \le \frac r2.
\]
Together with Proposition~\ref{prop:AC_consequences}(iii), this implies
\[
d(g_p^{(J)},m_{p,0})
\le d(g_p^{(J)},m_{p,J}) + d(m_{p,J},m_{p,0})
\le \tau + \frac r2 \le r,
\]
hence $g_p^{(J)}\in U_p$.

Now fix $p$.
Since $m_{p,J}$ maximizes $F_J$ on $U_p$ and $g_p^{(J)}\in U_p$,
\[
F_J(g_p^{(J)})\le F_J(m_{p,J}).
\]
Also $F_J$ is $(\mu/2)$-strongly concave and $M_J$-smooth on $U_p$,
so Lemma~\ref{lem:2sided} applied at \(g_0=m_{p,J}\) and \(g=g_p^{(J)}\) gives
\[
F_J(g_p^{(J)})
\ge F_J(m_{p,J})-\frac{M_J}{2}\tau^2.
\]
Choose any basin index $p^\star$ such that
\[
F_J(m_{p^\star,J})=\max_{q=1,\dots,P}F_J(m_{q,J}).
\]
Since every global maximizer of $F_J$ lies in $S$ by
Proposition~\ref{prop:AC_consequences}(iv), and since $m_{p,J}$ is the
unique maximizer of $F_J$ on $U_p$ for each $p$, we have
\[
\max_{g\in \SO(3)}F_J(g)
=
\max_{g\in S}F_J(g)
=
\max_{q=1,\dots,P}F_J(m_{q,J})
=
F_J(m_{p^\star,J}).
\]
Since $\hat g$ is the best final candidate,
\[
F_J(\hat g)
=
\max_{p=1,\dots,P} F_J(g_p^{(J)})
\ge F_J(g_{p^\star}^{(J)})
\ge F_J(m_{p^\star,J})-\frac{M_J}{2}\tau^2.
\]
Therefore
\[
0 \le \max_{g\in\SO(3)}F_J(g)-F_J(\hat g)\le \frac{M_J}{2}\tau^2.
\]
\end{proof}

Theorem~\ref{thm:multiband_correct} proves that the objective value at $\hat g$ is near globally optimal. To additionally guarantee that $\hat g$ itself is close to the globally
best maximizer requires an additional separation or curvature condition, such
as used in Corollary~\ref{cor:winning_basin}.

\begin{corollary}[Unique globally best basin is correctly identified]
\label{cor:winning_basin}\leavevmode\par\noindent
Assume \ref{ass:A-initialization}--\ref{ass:D-new_accuracy}, let
$p^\star \in \argmax_{p=1,\dots,P} F_J(m_{p,J})$, and set
\[
\gamma \eqdef
F_J(m_{p^\star,J})-\max_{p\neq p^\star}F_J(m_{p,J}).
\]
If $\frac{M_J}{2}\tau^2 < \gamma$, then $p^\star$ is the unique best basin and
\[
\hat p = p^\star,
\qquad
\hat g = g_{p^\star}^{(J)},
\qquad
d(\hat g,m_{p^\star,J})\le \tau.
\]
\end{corollary}

\begin{proof}
Note that $p^\star$ is the unique best basin: since $\frac{M_J}{2}\tau^2\ge 0$, and because $\frac{M_J}{2}\tau^2<\gamma$ forces $\gamma>0$, we have that $F_J(m_{p^\star,J})$ strictly exceeds $F_J(m_{p,J})$ for every
$p\neq p^\star$.

As in the proof of Theorem~\ref{thm:multiband_correct}, $g_p^{(J)}\in U_p$ for
every $p$. For the best basin,
\[
F_J(g_{p^\star}^{(J)})\ge F_J(m_{p^\star,J})-\frac{M_J}{2}\tau^2,
\]
while for every $p\neq p^\star$,
\[
F_J(g_p^{(J)})\le F_J(m_{p,J})\le F_J(m_{p^\star,J})-\gamma.
\]
Subtracting,
\[
F_J(g_{p^\star}^{(J)})-F_J(g_p^{(J)})\ge \gamma-\frac{M_J}{2}\tau^2>0,
\]
so the final $\argmax$ selects $p^\star$: $\hat p=p^\star$ and
$\hat g=g_{p^\star}^{(J)}$. Finally $d(\hat g,m_{p^\star,J})\le\tau$ by
Assumption~\ref{ass:D-new_accuracy}.
\end{proof}

\subsection{Benefits of marching}\label{sec:benefit_marching}

The assumptions of Theorem~\ref{thm:multiband_correct} are stated
in terms of quantities $a_j$, $b_j$, $\delta_j$.
To connect these with the bandwidth schedule, we use the fact that
the derivatives of bandlimited functions on $\SO(3)$ are bounded by the powers of the bandwidth. 

A natural question is whether intermediate bandwidths are beneficial and if so, by how much, compared with making a single big jump from $L_0$ to $L_J$. Indeed, frequency marching controls the perturbation induced by higher angular frequencies at the scale of each increment rather than the final bandwidth. 
A single jump from $L_0$ to $L_J$ requires the full perturbation $F_J-F_0$ to be small enough, relative to the basin curvature margin $\mu$ at $L_0$; as we show, marching distributes this perturbation over smaller increments.

Since each increment $E_j = F_j - F_{j-1}$ has bandwidth at most $L_j$, Lemma~\ref{lem:bernstein_SO3} yields
\begin{equation}\label{eq:bernstein_per_step}
    a_j \leq B_1 (1+L_j)\,\delta_j, \qquad b_j \leq B_2 (1+L_j)^2\,\delta_j,
\end{equation}
where we recall that $\delta_j = \|E_j\|_\infty$. This allows us to write the assumptions in terms of the quantities $(L_j, \delta_j)_{j=1}^J$.

We let $\lesssim$ absorb the Bernstein constants $B_1$, $B_2$ from Lemma~\ref{lem:bernstein_SO3}. Suppose we skip all intermediate levels and apply the basin persistence argument directly to the single perturbation $E_{\mathrm{tot}} = F_J - F_0$. Since $E_{\mathrm{tot}}$ has bandwidth $L_J$ we have by Lemma~\ref{lem:bernstein_SO3}
\[
    \|\nabla^2 E_{\mathrm{tot}}\|_\infty \lesssim (1+L_J)^2\,\|E_{\mathrm{tot}}\|_\infty.
\]
Preserving strong concavity of the tracked basins therefore is guaranteed by
\begin{equation}\label{eq:oneshot_requirement}
    \|E_{\mathrm{tot}}\|_\infty \lesssim \frac{\mu}{(1+L_J)^2}.
\end{equation}

By contrast, when marching, we need to apply Lemma~\ref{lem:bernstein_SO3} to each increment $E_j$ separately. Using~\eqref{eq:bernstein_per_step}, the  condition $\sum_j b_j \leq \mu/2$ from Assumption~(B) becomes
\[
    B_2 \sum_{j=1}^J (1+L_j)^2\,\delta_j \leq \frac{\mu}{2}.
\]
Write $\delta_{\mathrm{tot}} = \sum_j \delta_j$ and define the amplitude-weighted effective bandwidth for $\delta_{\mathrm{tot}} > 0$ as

\begin{equation}\label{eq:effective_bandwidth}
    \overline{L}_2^{\,2} \eqdef \frac{\sum_{j=1}^J (1+L_j)^2\,\delta_j}{\delta_{\mathrm{tot}}}\,
\end{equation}
so that the marching condition reads $\delta_{\mathrm{tot}} \lesssim \mu / \overline{L}_2^{\,2}$.
Note that $\delta_{\mathrm{tot}} \geq \|E_{\mathrm{tot}}\|_\infty$ by the triangle inequality, so the marching analysis works with a potentially larger numerator. The gain comes from the denominator: since $\overline{L}_2 \le 1+L_J$, with equality only when all nonzero increment sizes \( \delta_j \) occur at the final level, marching compensates by replacing the worst-case $(1+L_J)^2$ weight with the effective weight~$\overline{L}_2^{\,2}$. 
The denominator improves from $(1+L_J)^2$ to $\overline{L}_2^2$, but this can be offset by the larger numerator $\delta_{\mathrm{tot}} = \sum_j \delta_j$. Thus, marching is most advantageous when the increments are concentrated at lower bandwidths, which is the typical regime for smooth signals or noise-dominated high-frequency content. An analogous comparison holds for the gradient-based conditions. The one-shot approach requires $\|E_{\mathrm{tot}}\|_\infty \lesssim \mu r / L_J$, whereas marching imposes two per-step conditions: $L_j \delta_j \lesssim \mu r$ (to apply basin persistence at each level) and $\sum_j L_j \delta_j \lesssim \mu r$ (to control the cumulative drift of the tracked maxima). Both are easier to satisfy when the increments $\delta_j$ decay with~$j$, which is, again, a typical situation. 

\paragraph{A toy model}
We consider an example of uniform bandwidth steps $L_j = L_0 + j\Delta$ with geometrically decaying increments $\delta_j = \delta\,q^{j-1}$, $q \in (0,1)$. A direct computation in the $J \to \infty$ limit gives
\begin{equation}\label{eq:Lbar_toy}
    \overline{L}_2^{\,2}
    = (1+L_0 + \Delta)^2 + \frac{2(1+L_0 + \Delta)\Delta\,q}{1-q} + \frac{\Delta^2\,q(1+q)}{(1-q)^2}\,.
\end{equation}
The key observation is that this expression is \emph{independent of~$L_J$}: no matter how large the final bandwidth, the effective penalty depends only on~$L_0$ and the decay rate~$q$. When \(q\) is small enough that
\[
\frac{2(1+L_0 + \Delta)\Delta q}{1-q}
+
\frac{\Delta^2q(1+q)}{(1-q)^2}
\ll (1+L_0 + \Delta)^2,
\]
we have \(\overline{L}_2\approx 1+ L_0+ \Delta \), and the (a priori) marching requirement $\delta_{\mathrm{tot}} \lesssim \mu / (1+L_0+\Delta)^2$ is generically much milder than the one-shot requirement~\eqref{eq:oneshot_requirement}.
Thus, if most of the spectral change occurs in the first few bandwidth increments, the continuation difficulty is governed by those early bandwidths rather than by the final bandwidth.

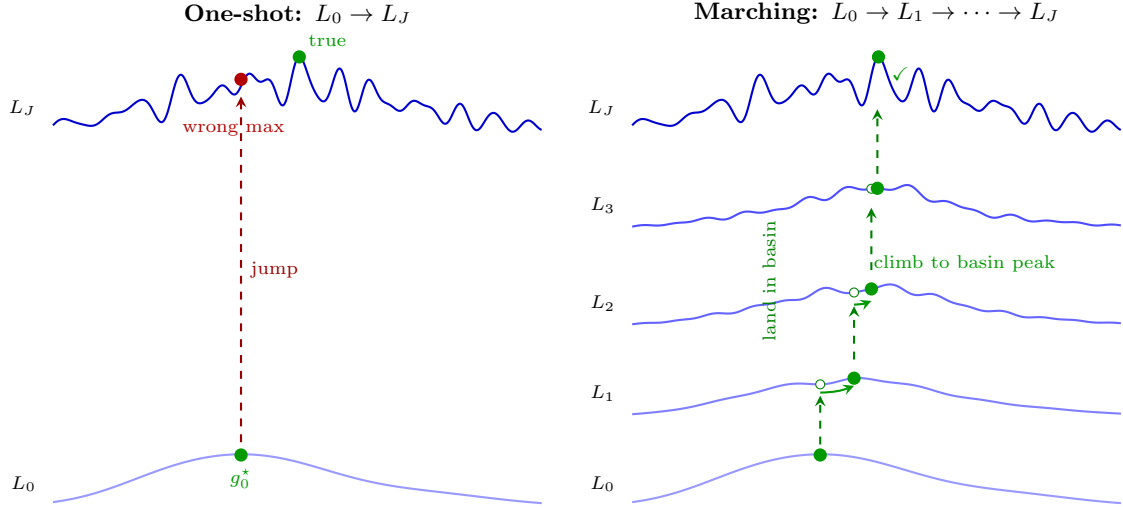
\begin{figure}[t]
\centering
\begin{tikzpicture}[>=stealth, scale=1.0]

% Parameters
\def\xmin{0}
\def\xmax{6.5}

% Common x-locations for tracked points
\def\xLzero{2.50}
\def\xLone{2.95}
\def\xLtwo{3.16}
\def\xLthree{3.25}
\def\xLJ{3.275}

% ============================================================
% LEFT PANEL: One-shot 
% ============================================================
\begin{scope}[shift={(0,0)}]

  \node[font=\small\bfseries, anchor=south] at (3.25, 6.3) {One-shot: $L_0 \to L_J$};

  % --- L_J: oscillatory landscape (top) ---
  \def\yoffJ{5.0}
  \draw[thick, blue!80!black] plot[domain=\xmin:\xmax, samples=400, smooth]
    (\x, {\yoffJ + 0.6*exp(-0.5*(\x-3.2)^2) + 0.25*exp(-0.8*(\x-1.5)^2) + 0.15*exp(-0.6*(\x-5.0)^2)
      + 0.18*sin(12*\x r)*exp(-0.15*(\x-3.2)^2)
      + 0.12*sin(8*\x r + 1.2)*exp(-0.2*(\x-2)^2)
      + 0.1*cos(15*\x r + 0.5)*exp(-0.1*(\x-4)^2)
      + 0.08*sin(20*\x r)*exp(-0.3*(\x-3.5)^2)});
  \node[font=\scriptsize, anchor=east] at (-0.1, \yoffJ+0.3) {$L_J$};

  % True maximum at L_J
  \fill[matchaColor] (\xLJ,
    {\yoffJ + 0.6*exp(-0.5*(\xLJ-3.2)^2) + 0.25*exp(-0.8*(\xLJ-1.5)^2) + 0.15*exp(-0.6*(\xLJ-5.0)^2)
      + 0.18*sin(12*\xLJ r)*exp(-0.15*(\xLJ-3.2)^2)
      + 0.12*sin(8*\xLJ r + 1.2)*exp(-0.2*(\xLJ-2)^2)
      + 0.1*cos(15*\xLJ r + 0.5)*exp(-0.1*(\xLJ-4)^2)
      + 0.08*sin(20*\xLJ r)*exp(-0.3*(\xLJ-3.5)^2)}) circle (2.5pt);
  \node[font=\scriptsize, matchaColor, anchor=south west] at (\xLJ, \yoffJ+1.0) {true};

  % Wrong local maximum reached from one-shot initialization
  \fill[red!70!black] (\xLzero,
    {\yoffJ + 0.6*exp(-0.5*(2.85-3.2)^2) + 0.25*exp(-0.8*(2.85-1.5)^2) + 0.15*exp(-0.6*(2.85-5.0)^2)
      + 0.18*sin(12*2.85 r)*exp(-0.15*(2.85-3.2)^2)
      + 0.12*sin(8*2.85 r + 1.2)*exp(-0.2*(2.85-2)^2)
      + 0.1*cos(15*2.85 r + 0.5)*exp(-0.1*(2.85-4)^2)
      + 0.08*sin(20*2.85 r)*exp(-0.3*(2.85-3.5)^2)}) circle (2.5pt);

%   % Arrow showing Newton going to wrong max
%   \draw[->, thick, red!60!black] (\xLzero, \yoffJ-0.15) to[bend left=20] (\xLzero, \yoffJ+0.45);
    \node[font=\scriptsize, red!70!black, anchor=south] at (\xLzero-0.1, \yoffJ-0.2) {wrong max};

  % --- Dashed arrow: big jump from L_0 to L_J ---
  \draw[dashed, thick, red!60!black, ->] (\xLzero, 0.85) -- (\xLzero, \yoffJ+0.45)
    node[midway, right, font=\scriptsize, red!60!black] {jump};

  % --- L_0: smooth landscape (bottom) ---
  \def\yoff{0.0}
  \draw[thick, blue!40] plot[domain=\xmin:\xmax, samples=200, smooth]
    (\x, {\yoff + 0.6*exp(-0.5*(\x-2.8)^2) + 0.25*exp(-0.8*(\x-1.5)^2) + 0.15*exp(-0.6*(\x-5.0)^2)});
  \node[font=\scriptsize, anchor=east] at (-0.1, \yoff+0.3) {$L_0$};

  % Mark the maximum at L_0
  \fill[matchaColor] (\xLzero, {\yoff + 0.68}) circle (2.5pt);
  \node[font=\scriptsize, matchaColor, anchor=north] at (\xLzero, {\yoff+0.62}) {$g_0^\star$};

\end{scope}

% ============================================================
% RIGHT PANEL: Marching
% ============================================================
\begin{scope}[shift={(7.7,0)}]

  \def\yoff{0.0}
  \node[font=\small\bfseries, anchor=south] at (3.25, 6.3) {Marching: $L_0 \to L_1 \to \cdots \to L_J$};

  % Common x-locations
  \def\xLzero{2.50}
  \def\xLone{2.95}
  \def\xLtwo{3.18}
  \def\xLthree{3.26}
  \def\xLJ{3.275}

  % --- L_0: smooth landscape (bottom) ---
  \def\yoffZero{0.0}
  \draw[thick, blue!40] plot[domain=\xmin:\xmax, samples=200, smooth]
    (\x, {\yoffZero + 0.6*exp(-0.5*(\x-2.8)^2) + 0.25*exp(-0.8*(\x-1.5)^2) + 0.15*exp(-0.6*(\x-5.0)^2)});
  \node[font=\scriptsize, anchor=east] at (-0.1, \yoffZero+0.3) {$L_0$};
  \fill[matchaColor] (\xLzero, {\yoffZero + 0.68}) circle (2.5pt);

  % --- L_1: broad basin, clearly rising to the right from the landing point ---
  \def\yoffA{1.2}
  \draw[thick, blue!50] plot[domain=\xmin:\xmax, samples=320, smooth]
    (\x, {\yoffA
      + 0.47*exp(-0.52*(\x-3.00)^2)
      + 0.14*exp(-0.85*(\x-1.55)^2)
      + 0.10*exp(-0.55*(\x-5.0)^2)
      - 0.055*exp(-10.0*(\x-2.62)^2)
      + 0.015*sin(7*\x r + 0.3)*exp(-0.30*(\x-3.0)^2)});
  \node[font=\scriptsize, anchor=east] at (-0.1, \yoffA+0.3) {$L_1$};

  % landing point on L_1
  \draw[matchaColor, fill=white] (\xLzero,
    {\yoffA
      + 0.47*exp(-0.52*(\xLzero-3.00)^2)
      + 0.14*exp(-0.85*(\xLzero-1.55)^2)
      + 0.10*exp(-0.55*(\xLzero-5.0)^2)
      - 0.055*exp(-10.0*(\xLzero-2.62)^2)
      + 0.015*sin(7*\xLzero r + 0.3)*exp(-0.30*(\xLzero-3.0)^2)}) circle (1.7pt);

  % peak on L_1
  \fill[matchaColor] (\xLone,
    {\yoffA
      + 0.47*exp(-0.52*(\xLone-3.00)^2)
      + 0.14*exp(-0.85*(\xLone-1.55)^2)
      + 0.10*exp(-0.55*(\xLone-5.0)^2)
      - 0.055*exp(-10.0*(\xLone-2.62)^2)
      + 0.015*sin(7*\xLone r + 0.3)*exp(-0.30*(\xLone-3.0)^2)}) circle (2.5pt);

  % --- L_2: landing from L_1 is visibly left of the local peak ---
  \def\yoffB{2.4}
  \draw[thick, blue!60] plot[domain=\xmin:\xmax, samples=340, smooth]
    (\x, {\yoffB
      + 0.51*exp(-0.65*(\x-3.18)^2)
      + 0.14*exp(-0.85*(\x-1.55)^2)
      + 0.10*exp(-0.55*(\x-5.0)^2)
      - 0.075*exp(-15.0*(\x-3.02)^2)
      + 0.030*sin(8*\x r + 0.5)*exp(-0.24*(\x-3.15)^2)
      + 0.020*sin(13*\x r - 0.2)*exp(-0.16*(\x-3.18)^2)});
  \node[font=\scriptsize, anchor=east] at (-0.1, \yoffB+0.3) {$L_2$};

  % landing point on L_2
  \draw[matchaColor, fill=white] (\xLone,
    {\yoffB
      + 0.51*exp(-0.65*(\xLone-3.18)^2)
      + 0.14*exp(-0.85*(\xLone-1.55)^2)
      + 0.10*exp(-0.55*(\xLone-5.0)^2)
      - 0.075*exp(-15.0*(\xLone-3.02)^2)
      + 0.030*sin(8*\xLone r + 0.5)*exp(-0.24*(\xLone-3.15)^2)
      + 0.020*sin(13*\xLone r - 0.2)*exp(-0.16*(\xLone-3.18)^2)}) circle (1.7pt);

  % peak on L_2
  \fill[matchaColor] (\xLtwo,
    {\yoffB
      + 0.51*exp(-0.65*(\xLtwo-3.18)^2)
      + 0.14*exp(-0.85*(\xLtwo-1.55)^2)
      + 0.10*exp(-0.55*(\xLtwo-5.0)^2)
      - 0.075*exp(-15.0*(\xLtwo-3.02)^2)
      + 0.030*sin(8*\xLtwo r + 0.5)*exp(-0.24*(\xLtwo-3.15)^2)
      + 0.020*sin(13*\xLtwo r - 0.2)*exp(-0.16*(\xLtwo-3.18)^2)}) circle (2.5pt);

  % --- L_3: small final drift to the right before L_J ---
  \def\yoffC{3.7}
  \draw[thick, blue!70] plot[domain=\xmin:\xmax, samples=360, smooth]
    (\x, {\yoffC
      + 0.54*exp(-0.82*(\x-3.25)^2)
      + 0.14*exp(-0.85*(\x-1.55)^2)
      + 0.10*exp(-0.55*(\x-5.0)^2)
      - 0.050*exp(-18.0*(\x-3.12)^2)
      + 0.040*sin(9*\x r + 0.45)*exp(-0.20*(\x-3.23)^2)
      + 0.028*sin(14*\x r - 0.2)*exp(-0.14*(\x-3.25)^2)});
  \node[font=\scriptsize, anchor=east] at (-0.1, \yoffC+0.3) {$L_3$};

  % landing point on L_3
  \draw[matchaColor, fill=white] (\xLtwo,
    {\yoffC
      + 0.54*exp(-0.82*(\xLtwo-3.25)^2)
      + 0.14*exp(-0.85*(\xLtwo-1.55)^2)
      + 0.10*exp(-0.55*(\xLtwo-5.0)^2)
      - 0.050*exp(-18.0*(\xLtwo-3.12)^2)
      + 0.040*sin(9*\xLtwo r + 0.45)*exp(-0.20*(\xLtwo-3.23)^2)
      + 0.028*sin(14*\xLtwo r - 0.2)*exp(-0.14*(\xLtwo-3.25)^2)}) circle (1.7pt);

  % peak on L_3
  \fill[matchaColor] (\xLthree,
    {\yoffC
      + 0.54*exp(-0.82*(\xLthree-3.25)^2)
      + 0.14*exp(-0.85*(\xLthree-1.55)^2)
      + 0.10*exp(-0.55*(\xLthree-5.0)^2)
      - 0.050*exp(-18.0*(\xLthree-3.12)^2)
      + 0.040*sin(9*\xLthree r + 0.45)*exp(-0.20*(\xLthree-3.23)^2)
      + 0.028*sin(14*\xLthree r - 0.2)*exp(-0.14*(\xLthree-3.25)^2)}) circle (2.5pt);

  % --- L_J: same endpoint as left panel ---
  \def\yoffJ{5.0}
  \draw[thick, blue!80!black] plot[domain=\xmin:\xmax, samples=400, smooth]
    (\x, {\yoffJ + 0.6*exp(-0.5*(\x-3.2)^2) + 0.25*exp(-0.8*(\x-1.5)^2) + 0.15*exp(-0.6*(\x-5.0)^2)
      + 0.18*sin(12*\x r)*exp(-0.15*(\x-3.2)^2)
      + 0.12*sin(8*\x r + 1.2)*exp(-0.2*(\x-2)^2)
      + 0.1*cos(15*\x r + 0.5)*exp(-0.1*(\x-4)^2)
      + 0.08*sin(20*\x r)*exp(-0.3*(\x-3.5)^2)});
  \node[font=\scriptsize, anchor=east] at (-0.1, \yoffJ+0.3) {$L_J$};

  % landing point on L_J
  \draw[matchaColor, fill=white] (\xLthree,
    {\yoffJ + 0.6*exp(-0.5*(\xLthree-3.2)^2) + 0.25*exp(-0.8*(\xLthree-1.5)^2) + 0.15*exp(-0.6*(\xLthree-5.0)^2)
      + 0.18*sin(12*\xLthree r)*exp(-0.15*(\xLthree-3.2)^2)
      + 0.12*sin(8*\xLthree r + 1.2)*exp(-0.2*(\xLthree-2)^2)
      + 0.1*cos(15*\xLthree r + 0.5)*exp(-0.1*(\xLthree-4)^2)
      + 0.08*sin(20*\xLthree r)*exp(-0.3*(\xLthree-3.5)^2)}) circle (1.7pt);

  % true maximum at L_J
  \fill[matchaColor] (\xLJ,
    {\yoffJ + 0.6*exp(-0.5*(\xLJ-3.2)^2) + 0.25*exp(-0.8*(\xLJ-1.5)^2) + 0.15*exp(-0.6*(\xLJ-5.0)^2)
      + 0.18*sin(12*\xLJ r)*exp(-0.15*(\xLJ-3.2)^2)
      + 0.12*sin(8*\xLJ r + 1.2)*exp(-0.2*(\xLJ-2)^2)
      + 0.1*cos(15*\xLJ r + 0.5)*exp(-0.1*(\xLJ-4)^2)
      + 0.08*sin(20*\xLJ r)*exp(-0.3*(\xLJ-3.5)^2)}) circle (2.5pt);
  \node[font=\scriptsize, matchaColor, anchor=south west] at (\xLJ, \yoffJ+1.0) {true};

  % Vertical "land in basin" arrows
  \draw[->, thick, matchaColor, dashed] (\xLzero, 0.82) -- (\xLzero, 1.47);
  \draw[->, thick, matchaColor, dashed] (\xLone, 1.95) -- (\xLone, 2.66);
  \draw[->, thick, matchaColor, dashed] (\xLtwo, 3.14) -- (\xLtwo, 3.97);
  \draw[->, thick, matchaColor, dashed] (\xLthree, 4.42) -- (\xLthree, 5.28);

  % Within-level climb to basin peak
  \draw[->, thick, matchaColor] (\xLzero, 1.51) to[bend right=10] (\xLone, 1.6);
  \draw[->, thick, matchaColor] (\xLone, 2.68) to[bend right=10] (\xLtwo, 2.73);
 
  \node[font=\scriptsize, matchaColor, rotate=90, anchor=south] at (2.03, 2.955) {land in basin};
  \node[font=\scriptsize, matchaColor, anchor=south] at (4.52, 3.0) {climb to basin peak};
 % Mark the maximum at L_0
  \fill[matchaColor] (\xLzero, {\yoff + 0.68}) circle (2.5pt);
  \node[font=\scriptsize, matchaColor, anchor=north] at (\xLzero, {\yoff+0.62}) {$g_0^\star$};
\end{scope}
\end{tikzpicture}
\caption{One-shot versus frequency marching on a schematic 1D correlation landscape.
  \textbf{Left:} jumping directly from~$L_0$ to~$L_J$, Newton's
  method \emph{can} converge to a spurious local maximum that
  emerged at high bandwidth.
  \textbf{Right:} under the conditions of
  Theorem~\ref{thm:multiband_correct}, marching through intermediate
  bandwidths keeps the tracked maximizer in the correct basin at
  each level; Newton refines within the basin before proceeding to
  the next bandwidth.}
\end{figure}

%%%% EXPERIMENTS %%%%% 
\section{Experiments} \label{sec:experiments}
We first validate the proposed method on synthetic data, demonstrating substantial gains in both speed and accuracy over exhaustive rotational search.
We then show that our approach can be used as a drop-in replacement in the subtomogram averaging (STA) pipeline of RELION, accelerating rotational alignment by an order of magnitude while refining a ribosome to $4$\angstrom{}, reaching the local Nyquist limit of $3.8$\angstrom{}.

The experiments were conducted either on one A100-80GB GPU (Section~\ref{sec:exp-synthetic}) or  on four RTX4090 GPUs (Section~\ref{sec:exp_STA}). They can be reproduced using code available at \href{https://github.com/swing-research/Matcha}{the Matcha repository}.

\subsection{Rotational search: synthetic examples}  \label{sec:exp-synthetic}

\def\ps{0.32}
\def\sz{4cm}
\begin{figure}
	\centering
	\begin{subfigure}[t]{\ps\textwidth}
		\centering
		\includegraphics[height=\sz]{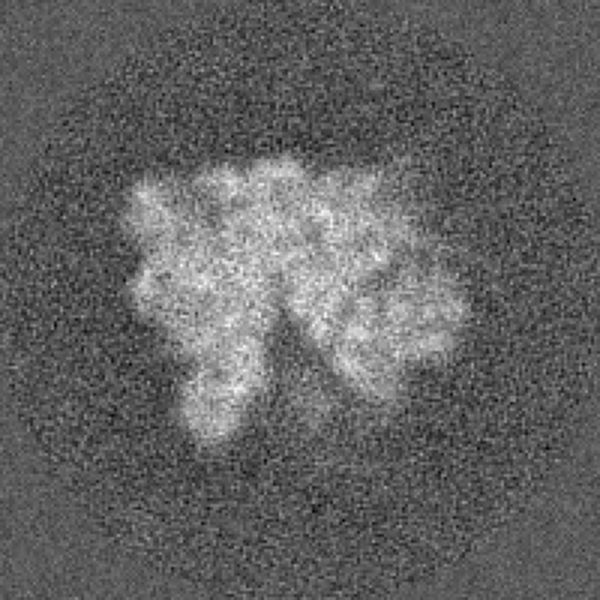}
		\caption{$\text{SNR}=20\text{dB}$.}
	\end{subfigure}%\hfill
	\begin{subfigure}[t]{\ps\textwidth}
		\centering
		\includegraphics[height=\sz]{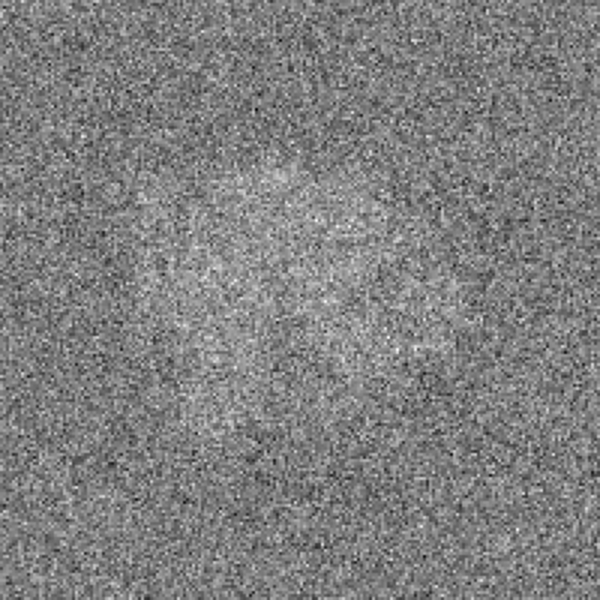}
		\caption{$\text{SNR}=0\text{dB}$.}
	\end{subfigure}%\hfill
	\begin{subfigure}[t]{\ps\textwidth}
		\centering
		\includegraphics[height=\sz]{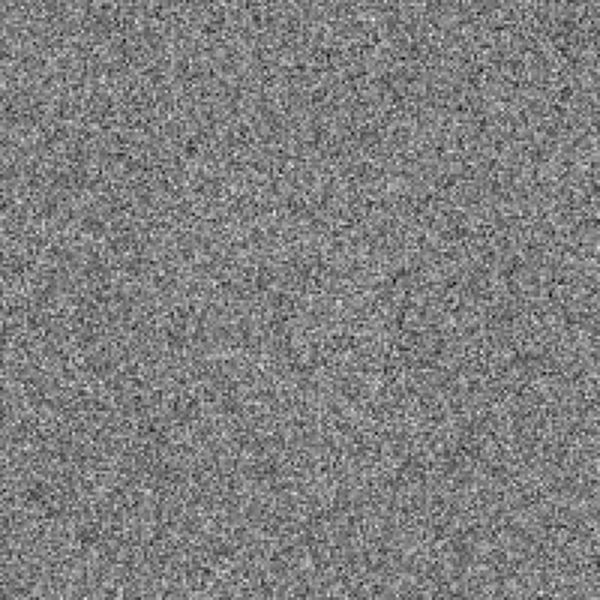}
		\caption{$\text{SNR}=-20\text{dB}$.}
	\end{subfigure}%\hfill
\caption{Slice of a ribosome volume corrupted by Gaussian noise at the indicated SNR levels.}\label{fig:exp_snr}
\end{figure}

Let $h \in \mathbb{R}^{N \times N \times N}$ denote a ribosome density 
at $13$\angstrom{} resolution (EMDB: EMD-3228), obtained from the EMPIAR-10045 dataset\footnote{\href{https://www.ebi.ac.uk/emdb/EMD-3228}{https://www.ebi.ac.uk/emdb/EMD-3228}}.
We synthetically generate measurements
\begin{equation*}
f = g^\star\circ  h + \eta,    
\end{equation*}
where $g^\star \in \SO(3)$ is the ground truth rotation (assumed unknown in the search) and $\eta$ is additive white Gaussian noise. We consider three signal-to-noise ratios: $20~\text{dB}$, $0~\text{dB}$, and $-20~\text{dB}$. Figure~\ref{fig:exp_snr} shows a representative volume slice at each level; at $0~\text{dB}$ the signal and noise have equal energy, and single-particle cryo-EM and cryo-ET alignment often operate in very low-SNR regimes, motivating the \(-20\) dB stress test. We report the results at $0~\text{dB}$ in the main text and defer the remaining SNR levels to Appendix~\ref{sec:other_noise}.

At each SNR, we generate $1000$ randomly rotated volumes of size $N=200$ and compare the performance of Algorithm~\ref{algo:coarse-to-fine} with  SOFFT~\cite{kostelec_ffts_2008,potts_fast_2009,mcewen2015novel}. SOFFT evaluates the cross-correlation on a discrete grid of rotations whose resolution is controlled by an oversampling factor~$K$: increasing \(K\) refines the grid but increases the number of sampled rotations by a factor of order \(K^3\).
For SOFFT experiments we used a batch size of 10 for \(K\le 8\), 4 for \(K=10\), and 1 for \(K=18\), as larger batches exceed GPU memory at higher bandwidths; even so, SOFFT frequently exceeds the 80GB memory limit at large oversampling factors (see Figure~\ref{fig:accuracy-v-bandlimit-snr0db}).

For \method{}, we select $10$ candidate rotations as the highest local maxima of the cross-correlation computed by SOFFT at a coarse bandwidth $L_0=30$ with oversampling factor~$2$.
Candidates are then refined across bandwidth $L_j \in \{30,40,60,L_{\max}\}$, skipping any $L_j$ that exceed $L_{\max}$: at each~$L_j$, a single Newton step is applied to every candidate, and refined estimates are propagated to the next band. After optimization at $L_{\max}$, the candidate with the highest correlation score is retained.

Figure~\ref{fig:accuracy-v-time} reports the wall-clock time and 90th-percentile alignment error for a fixed bandwidth $L_{\max} = 40$, which yields the fastest inference but not the best resolution---especially for SOFFT, whose accuracy is fundamentally limited by its grid spacing. At the highest oversampling factor that does not cause out-of-memory errors ($K = 18$), SOFFT achieves $0.13^\circ$  accuracy but requires over a hundred seconds. For $K=8$ and above, wall-clock times range from ten seconds to several minutes. By contrast, \method{} achieves $0.03^\circ$ accuracy---over four times more precise---within a few seconds, despite starting from a coarse initial grid. 

This gain stems from a structural advantage: unlike SOFFT, whose accuracy is fundamentally limited by the grid spacing at a given bandwidth, \method{} uses the grid only to identify candidate basins and then refines continuously to a precision that is independent of the initial grid resolution.

\method{} also uses memory more efficiently, which enables larger batch sizes when processing multiple volumes. Under identical hardware constraints, SOFFT is limited to a batch size of 10 for $K \le 8$, a batch size of $4$ for $K = 10$, and a batch size of $1$ for $K = 18$, whereas \method{} supports batch sizes of up to $100$.

The maximum bandwidth $L_{\max}$ is a key hyperparameter of both methods: small values yield fast computation but introduce bias that limits accuracy, while larger values are slower but more accurate. Figure~\ref{fig:accuracy-v-bandlimit-snr0db} confirms that higher $L_{\max}$ consistently improves estimation accuracy. The effect is moderate for \method{} and more pronounced for SOFFT at all oversampling factors.  Beyond $L_{\max} = 100$, improvements for \method{} are marginal, whereas SOFFT becomes infeasible due to memory constraints.

We also compare the Newton refinement used in \method{} with a first-order gradient method. Since both methods require evaluating the same bandlimited objective, the practical difference is the number of refinement steps needed to reach high precision. Figure~\ref{fig:newton_v_gd} shows that Newton refinement reaches close to optimal accuracy after very few steps, whereas the first-order method requires many more steps for the desired accuracy. This supports our use of one Newton step per bandwidth level in the synthetic experiments.
\newlength{\panelsep}
\setlength{\panelsep}{4em}

\newlength{\panelwidth}
\setlength{\panelwidth}{\dimexpr(\textwidth-\panelsep)/2\relax}

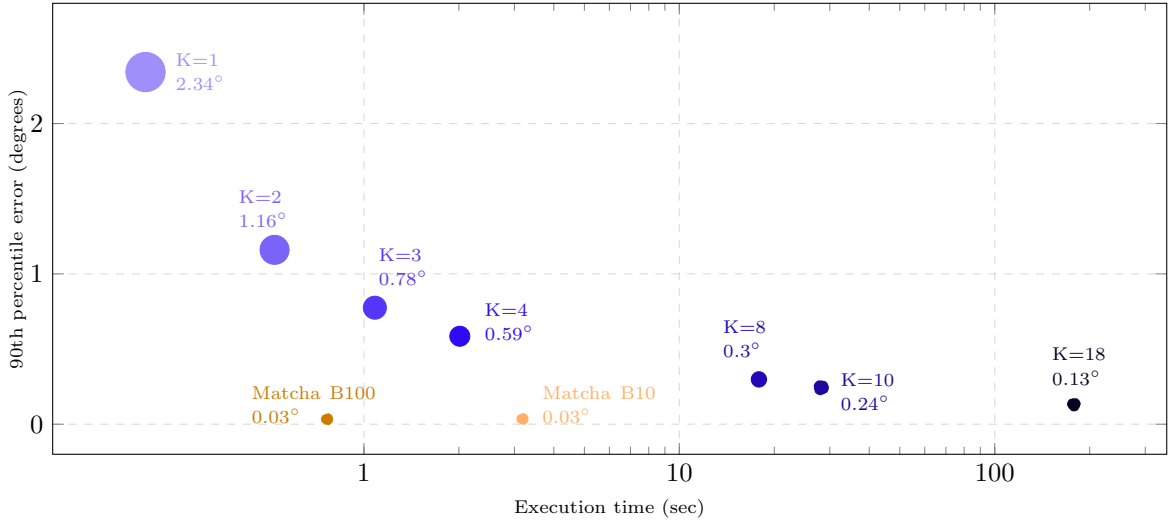
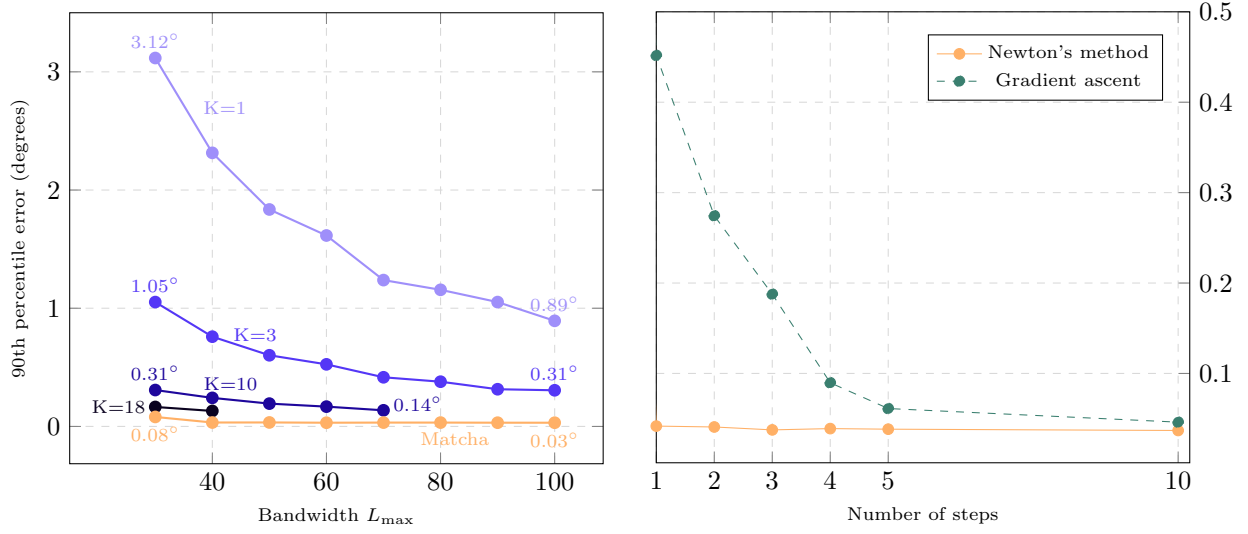
\begin{figure}[htbp]
    \centering
     \hspace{-0.01\textwidth}\begin{subfigure}[t]{0.95\textwidth}
    \centering
    \begin{tikzpicture}
        \begin{axis}[
            xmode=log,
            xlabel={\scriptsize Execution time (sec)},
            ylabel={\scriptsize90th percentile error (degrees)},
            ymin=-0.2,
            ymax=2.8,
            xtick={0.001, 0.01,0.1,1},
            xticklabels={1,10,100,1000},
         	grid=major, % Display a grid
			grid style={dashed,gray!30}, % Set the style
            width=\linewidth,
            height=6cm,
            %ticklabel style={font=\scriptsize},
            ylabel style={yshift=-0.2cm},
            scale only axis,
            trim axis left,
            trim axis right,
        ]
        
\newcommand{\snr}{snr0}
\newcommand{\pathPlotDots}{plots/plot2_newlayout/}

\DTLloaddb{mydata}{\pathPlotDots/sofft_K1_\snr.csv}
\DTLgetvalue{\alphavalA}{mydata}{1}{2}
\pgfmathparse{round(\alphavalA*100)/100}
\edef\alphavalroundedA{\pgfmathresult}
\pgfmathparse{(\alphavalroundedA*1.5+0.2)}
\edef\alphascaleA{\pgfmathresult}
\addplot+[mark=*, color=cOv1, mark options={scale=\alphascaleA}]
  table[col sep=comma, x=x, y=y] {\pathPlotDots/sofft_K1_\snr.csv} node[pos=0,right, text width=5cm] {\scriptsize{\quad K=1 \\\vspace{-0.1cm}\quad$\pgfmathprintnumber[fixed,precision=2]{\alphavalA}^\circ$}};
\DTLgetvalue{\xA}{mydata}{1}{1}
\DTLgetvalue{\yA}{mydata}{1}{2}
\DTLcleardb{mydata}
\DTLdeletedb{mydata}

\DTLloaddb{mydata}{\pathPlotDots/sofft_K2_\snr.csv}
\DTLgetvalue{\alphavalB}{mydata}{1}{2}
\pgfmathparse{round(\alphavalB*100)/100.}
\edef\alphavalroundedB{\pgfmathresult}
\pgfmathparse{(\alphavalroundedB*1.5 + 1)}
\edef\alphascaleB{\pgfmathresult}
\addplot+[mark=*, color=cOv2, mark options={scale=\alphascaleB}]
  table[col sep=comma, x=x, y=y] {\pathPlotDots/sofft_K2_\snr.csv} node[pos=0,text width=5cm, anchor=south west, xshift= -25pt, yshift=5pt] {\scriptsize{\quad K=2 \\\vspace{-0.1cm}\quad$\pgfmathprintnumber[fixed,precision=2]{\alphavalB}^\circ$}};
\DTLcleardb{mydata}
\DTLdeletedb{mydata}

\DTLloaddb{mydata}{\pathPlotDots/sofft_K3_\snr.csv}
\DTLgetvalue{\alphavalC}{mydata}{1}{2}
\pgfmathparse{round(\alphavalC*100)/100}
\edef\alphavalroundedC{\pgfmathresult}
\pgfmathparse{(\alphavalroundedC*1.5 + 1)}
\edef\alphascaleC{\pgfmathresult}
\addplot+[mark=*, color=cOv3, mark options={scale=\alphascaleC}]
  table[col sep=comma, x=x, y=y] {\pathPlotDots/sofft_K3_\snr.csv} node[pos=0,right, text width=5cm, anchor=south west, yshift=5pt, xshift= -10pt] {\scriptsize{\quad K=3 \\\vspace{-0.1cm}\quad$\pgfmathprintnumber[fixed,precision=2]{\alphavalC}^\circ$}};
\DTLcleardb{mydata}
\DTLdeletedb{mydata}

\DTLloaddb{mydata}{\pathPlotDots/sofft_K4_\snr.csv}
\DTLgetvalue{\alphavalD}{mydata}{1}{2}
\pgfmathparse{round(\alphavalD*100)/100}
\edef\alphavalroundedD{\pgfmathresult}
\pgfmathparse{(\alphavalroundedD*1.5+ 1)}
\edef\alphascaleD{\pgfmathresult}
\addplot+[mark=*, color=cOv5, mark options={scale=\alphascaleD}]
  table[col sep=comma, x=x, y=y] {\pathPlotDots/sofft_K4_\snr.csv} node[pos=0,right, text width=5cm, anchor=south west, yshift=-5pt, xshift=-2pt]{\scriptsize{\quad K=4 \\\vspace{-0.1cm}\quad$\pgfmathprintnumber[fixed,precision=2]{\alphavalD}^\circ$}};
\DTLcleardb{mydata}
\DTLdeletedb{mydata}

\DTLloaddb{mydata}{\pathPlotDots/sofft_K8_\snr.csv}
\DTLgetvalue{\alphavalE}{mydata}{1}{2}
\pgfmathparse{round(\alphavalE*100)/100}
\edef\alphavalroundedE{\pgfmathresult}
\pgfmathparse{(\alphavalroundedE*1.5 + 1)}
\edef\alphascaleE{\pgfmathresult}
\addplot+[mark=*, color=cOv8, mark options={scale=\alphascaleE}]
  table[col sep=comma, x=x, y=y] {\pathPlotDots/sofft_K8_\snr.csv} node[pos=0,right, text width=5cm,anchor=south west, yshift=5pt, xshift=-25pt]{\scriptsize{\quad K=8\\\vspace{-0.1cm}\quad $\pgfmathprintnumber[fixed,precision=2]{\alphavalE}^\circ$}};
\DTLcleardb{mydata}
\DTLdeletedb{mydata}

\DTLloaddb{mydata}{\pathPlotDots/sofft_K10_\snr.csv}
\DTLgetvalue{\alphavalF}{mydata}{1}{2}
\pgfmathparse{round(\alphavalF*100)/100}
\edef\alphavalroundedF{\pgfmathresult}
\pgfmathparse{(\alphavalroundedF*1.5 + 1)}
\edef\alphascaleF{\pgfmathresult}
\addplot+[mark=*, color=cOv10, mark options={scale=\alphascaleF}]
  table[col sep=comma, x=x, y=y] {\pathPlotDots/sofft_K10_\snr.csv} node[right, text width=5cm,anchor=south west, yshift=-12pt, xshift=-4pt]{\scriptsize{\quad K=10\\\vspace{-0.1cm}\quad$\pgfmathprintnumber[fixed,precision=2]{\alphavalF}^\circ$}};
\DTLcleardb{mydata}
\DTLdeletedb{mydata}

\DTLloaddb{mydata}{\pathPlotDots/sofft_K18_\snr.csv}
\DTLgetvalue{\alphavalL}{mydata}{1}{2}
\pgfmathparse{round(\alphavalL*100)/100}
\edef\alphavalroundedL{\pgfmathresult}
\pgfmathparse{(\alphavalroundedL*1.5+ 1)}
\edef\alphascaleL{\pgfmathresult}
\addplot+[mark=*, color=cOv18, mark options={scale=\alphascaleL}]
  table[col sep=comma, x=x, y=y] {\pathPlotDots/sofft_K18_\snr.csv} node[pos=0,below, text width=5cm,anchor=north, xshift=55pt, yshift=25pt] {\scriptsize{\quad K=18\\\vspace{-0.1cm}\quad$\pgfmathprintnumber[fixed,precision=2]{\alphavalL}^\circ$}};
\DTLcleardb{mydata}
\DTLdeletedb{mydata}

\DTLloaddb{mydata}{\pathPlotDots/matcha_bs10_steps1_cand10_\snr.csv}
\DTLgetvalue{\alphavalH}{mydata}{1}{2}
\pgfmathparse{round(\alphavalH*100)/100}
\edef\alphavalroundedH{\pgfmathresult}
\pgfmathparse{(\alphavalroundedH*1.5 + 1)}
\edef\alphascaleH{\pgfmathresult}
\addplot+[mark=*, color=cOurs10, mark options={scale=\alphascaleH}]
  table[col sep=comma, x=x, y=y] {\pathPlotDots/matcha_bs10_steps1_cand10_\snr.csv} node[pos=0,right, text width=5cm,anchor=south west, yshift=-5pt, xshift=-4pt] {\scriptsize{\quad \method{} B10\\\vspace{-0.1cm}\quad $\pgfmathprintnumber[fixed,precision=2]{\alphavalH}^\circ$}};
\DTLcleardb{mydata}
\DTLdeletedb{mydata}

\DTLloaddb{mydata}{\pathPlotDots/matcha_bs100_steps1_cand10_\snr.csv}
\DTLgetvalue{\alphavalI}{mydata}{1}{2}
\pgfmathparse{round(\alphavalI*100)/100}
\edef\alphavalroundedI{\pgfmathresult}
\pgfmathparse{(\alphavalroundedI*1.5 + 1)}
\edef\alphascaleI{\pgfmathresult}
\addplot+[mark=*, color=cOurs100, mark options={scale=\alphascaleI}]
  table[col sep=comma, x=x, y=y] {\pathPlotDots/matcha_bs100_steps1_cand10_\snr.csv} node[pos=0, above, text width=5cm,anchor=south west, yshift=-5pt, xshift=-40pt] {\scriptsize{\quad \method{} B100\\\vspace{-0.1cm}\quad $\pgfmathprintnumber[fixed,precision=2]{\alphavalI}^\circ$}};
\DTLcleardb{mydata}
\DTLdeletedb{mydata}
        \end{axis}
    \end{tikzpicture}
    
    \caption{Accuracy vs.\ execution time ($L_{\max}=40$).}
        \label{fig:accuracy-v-time}
\end{subfigure}    
\medskip
    \begin{subfigure}[t]{\panelwidth}
    \centering
            \begin{tikzpicture}
         \begin{axis}[
            % ymode=log,
            xlabel={\scriptsize Bandwidth $L_\mathrm{max}$},
            ylabel={\scriptsize 90th percentile error (degrees)},
            ymax=3.5,
            width=\linewidth,
			grid=major, % Display a grid
			grid style={dashed,gray!30}, % Set the style
            height=6cm,
            xmin = 15,
            xtick={40,60,80,100},
            scale only axis,
            trim axis left,
            trim axis right,
        ]
            \newcommand{\snr}{snr0}
            
            \PlotCurveWithLabels{plots/plot1_newlayout/ov1_\snr.csv}{cOv1}{\hspace{0.7cm} K=1}{solid, thick}{above}{0.1}{above}{0pt}{0pt}
           
            \PlotCurveWithLabels{plots/plot1_newlayout/ov3_\snr.csv}{cOv3}{K=3}{solid, thick}{above}{0.25}{above}{0pt}{0pt}
            \PlotCurveWithLabels{plots/plot1_newlayout/ov10_\snr.csv}{cOv10}{K=10}{solid, thick}{above}{0.33}{above}{13pt}{-4pt}
            \PlotCurveWithLabelsSecond{plots/plot1_newlayout/ov18_\snr.csv}{cOv18}{K=18}{solid, thick}{}{0.0}
            \PlotCurveWithLabels{plots/plot1_newlayout/ours_\snr.csv}{cOurs}{\method{}}{solid, thick}{below}{0.75}{below}{0pt}{0pt}
        \end{axis}
    \end{tikzpicture}
   
    \caption{Effect of $L_{\max}$ for varying oversampling $K$.}
    \label{fig:accuracy-v-bandlimit-snr0db}
    \end{subfigure}
    \hfill
    \begin{subfigure}[t]{\panelwidth}
    \centering
    \begin{tikzpicture}
        \begin{axis}[
            xlabel={\scriptsize Number of steps},
            ymin=0.0012,
            ymax=0.5,
            xtick={0,1,2,3,4,5,10},
            xmin=1,
            xmax=10.2,
            width=\linewidth,
            %ticklabel style={font=\scriptsize},
            height=6cm,
            legend style={legend columns=1,at={(0.95,0.95)}},
			grid=major, % Display a grid
			grid style={dashed,gray!30}, % Set the style
            ylabel style={yshift=-0.2cm},
            scale only axis,
            trim axis right,
            trim axis left,
            ytick pos=right,
            yticklabel pos=right,
        ]
            \PlotCurveWithLast{plots/plot3/newton_snr0.csv}{cOurs10}{}{}{solid}{below}{0.8}
            \PlotCurveWithLast{plots/plot3/gd_snr0.csv}{cGD}{}{}{dashed}{above}{0.8}
            \addlegendimage{solid, cOurs10, mark=*}
            \addlegendentry{\scriptsize  Newton's method}
            \addlegendimage{dashed, cGD, mark=*}
            \addlegendentry{\scriptsize Gradient ascent}
        \end{axis}
    \end{tikzpicture}
    \caption{Newton's method vs.\ gradient ascent ($L_{\max}=40$).}
    \label{fig:newton_v_gd}
\end{subfigure}

\caption{Hyperparameter study for \method{} and SOFFT at $\mathrm{SNR}=0\,\mathrm{dB}$. In~(a), dot size encodes alignment error; \method{} achieves $0.03^\circ$ accuracy in seconds, whereas SOFFT at its highest feasible oversampling ($K=18$) requires over a hundred seconds for $0.13^\circ$. In~(b), increasing $L_{\max}$ improves both methods, but the effect is more pronounced for SOFFT; beyond $L_{\max}=100$, SOFFT becomes memory-prohibitive at large~$K$. In~(c), a single Newton step ($M_{\mathrm{iter}} = 1)$ matches the accuracy that gradient ascent reaches only after several iterations, supporting the use of one Newton step per band in~(a) and~(b).}
\end{figure}

\subsection{Nyquist resolution in cryo-ET using STA}\label{sec:exp_STA}

We integrate \method{} into the subtomogram averaging pipeline of RELION-5~\cite{burt2024image} as a replacement for the pose-refinement step performed during \emph{3D auto-refine}. 
All other processing steps, including classification, CTF refinement, Bayesian polishing or tomography refinement, reconstruction, and post-processing, are carried out using RELION's own implementation. 
Because these steps are shared between the baseline and our method, we exclude them from the runtime comparison and report only the cumulative wall-clock time spent on alignment.

We follow the STA tutorial\footnote{\href{https://tomoguide.github.io}{https://tomoguide.github.io}} for the Chlamy dataset (EMPIAR-11830), originally published by~\cite{kelley2026toward}. 
Particles were picked by template matching using pytom-match-pick, following the TomoGuide workflow. 
The RELION baseline refines progressively through three binning levels (bin4, bin2, bin1), reaching \(3.8\,\angstrom\) local resolution at bin1---the Nyquist limit for voxel size \(1.9\,\angstrom\)---under gold-standard FSC at a threshold of \(0.143\). 
The cumulative RELION alignment time across all three stages is approximately \(3.5\)~hours on four RTX\,4090 GPUs. The implementation of \method{} operates on particles extracted as three-dimensional pseudo-subtomograms.

The initial bin4 refinement is performed entirely in RELION, not in \method{}. 
For this bin4 stage, we use the orientations determined by pytom-match-pick as initialization for the local RELION search, as in the tutorial. 
At the bin2 and bin1 stages, the RELION baseline \emph{3D auto-refine} uses local angular searches initialized from the previous refinement stage. 
We replace these bin2 and bin1 rotational refinement stages by \method{}. 
The current implementation of \method{} performs a global rotational search at each replacement stage, rather than a local search around the previous orientation. 
This makes the comparison conservative in RELION's favor, since \method{} solves a larger rotational search problem at bin2 and bin1.

\method{} estimates both rotations and translations using the same half-set split as RELION, with translations recovered by FFT-based correlation; details are given in Appendix~\ref{sec:appendix-6d-search}. 
At each binning level, \method{} performs \(T=3\) outer rotation--translation alternations. 
Within each alternation, we retain \(P=10\) rotation candidates from a global low-bandwidth SOFFT search at angular-degree cutoff \(L_0=30\) with oversampling factor~\(2\), and refine them with \(M_{\mathrm{iter}}=5\) Newton steps through bandwidths \(\{30,40,L_{\max}\}\), where \(L_{\max}=45\) at bin2 and \(L_{\max}=60\) at bin1. 
This is more conservative than the \(1\)--\(2\) Newton steps sufficient in the synthetic setting (Section~\ref{sec:exp-synthetic}), but improves robustness under the noisier conditions of experimental data. 
Shifts are estimated by FFT-based correlation with subpixel refinement~\cite{guizar-sicairos_efficient_2008}.

Our bin1 reconstruction attains \(3.8\,\angstrom \) local resolution, matching the RELION baseline at the Nyquist limit, and \(4.0\,\angstrom\) global resolution under the same FSC criterion. 
The cumulative alignment time drops from approximately \(3.5\)~hours with RELION to roughly \(20\)~minutes with \method{}, a speedup of over \(10\times\) (Figure~\ref{fig:emp_11830_timings}). 
The non-alignment processing steps add a fixed overhead of approximately \(90\)~minutes to both pipelines, so the end-to-end speedup is smaller than the alignment-only speedup.
\begin{figure}
    \centering
     \includegraphics[width=0.8\linewidth]{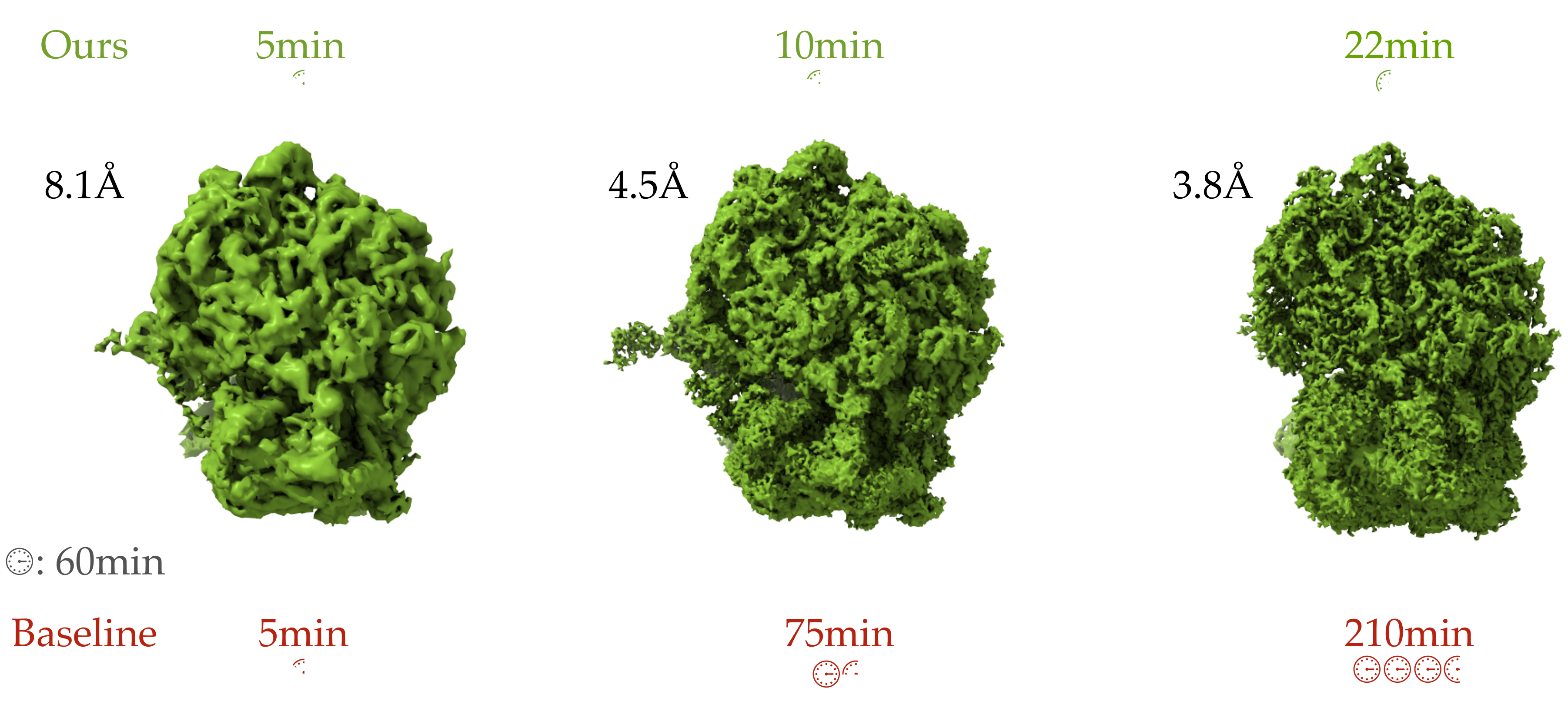}
    \caption{Subtomogram averaging on the Chlamy dataset (EMPIAR-11830) using RELION-5. \method{} replaces \emph{3D auto-refine} at the bin2 and bin1 stages, matching the baseline resolution ($3.8$\,\angstrom{} at Nyquist) while reducing cumulative alignment time by over $10\times$. Intermediate processing steps (CTF-refinement, classification, template updates) are shared between both pipelines and excluded from the reported timings.}
    \label{fig:emp_11830_timings}   
\end{figure}

\newpage
%%%% Conclusion %%%%% 
\section{Conclusion}
We presented \method{}, a frequency-marching algorithm for fast rotational alignment that replaces exhaustive high-resolution search with coarse SOFFT initialization followed by Newton refinement across progressively finer angular bandwidths. 
Algorithmically, the method works because exhaustive search is performed only at low angular bandwidth, while subsequent high-bandwidth refinement is continuous and low-dimensional. 
Theoretically, the frequency-marching analysis shows that continuation through intermediate bandwidths replaces a worst-case final-bandwidth perturbation requirement by an effective bandwidth determined by the sizes of the increments.

On synthetic benchmarks, \method{} achieves sub-degree angular accuracy with one-to-two-order-of-magnitude speedups over exhaustive search in the tested settings. 
Integrated into the RELION-5 subtomogram-averaging pipeline, it reduces alignment runtime by over \(10\times\) while matching Nyquist-limited reconstruction quality, demonstrating that the approach is practical for large-scale cryo-ET datasets.

Several directions remain open. 
First, our convergence analysis is deterministic and assumes spectral-decay and gap conditions; establishing probabilistic guarantees for realistic noisy-particle ensembles would be valuable. 
Second, the current alternating rotation--translation scheme inherits the limitations of block coordinate ascent. 
A joint continuous optimization over the full six-dimensional pose space, exploiting the semidirect-product structure of \(\SE(3)\), could improve both accuracy and convergence speed.

Third, the present rotational objective does not explicitly include particle-specific Fourier-space weighting during the Wigner--\(D\) rotational matching step. 
In RELION's 3D subtomogram workflow, each particle may be associated with a Fourier-space CTF weight volume; during alignment, the reference can be rotated and then weighted in Fourier space before scoring. 
In \method{}, by contrast, the rotational correlation is expanded once in a ball-harmonic/Wigner--\(D\) representation, after which rotations are evaluated implicitly through Wigner--\(D\) matrices. 
This makes rotational scoring efficient, but makes particle-specific, rotation-dependent CTF or missing-information weighting less direct. 
Weighted correlation objectives, including missing-wedge-aware variants, are possible in principle~\cite{chen_fast_2013}, but incorporating them while preserving the efficiency of the precomputed Wigner--\(D\) expansion remains future work.

Finally, in practice the reference template is itself unknown and must be iteratively estimated from the data, and the particle population may contain multiple conformational states. 
The speed of \method{} makes it feasible to evaluate alignments against multiple candidate templates per iteration, opening a path toward faster heterogeneity analysis in subtomogram averaging.
\section*{Acknowledgments}
We thank Ricardo Diogo Righetto for extensive feedback on the manuscript and for helpful discussions on the subtomogram-averaging experiments and implementation details. 
We thank Florent Waltz for help with the ribosome tutorial and related RELION/TomoGuide workflow details.
We thank Lorenzo Baldassari for valuable feedback on the manuscript. 
Valentin Debarnot is supported by the Agence Nationale de la Recherche (ANR) and the Ministère de l’Enseignement Supérieur et de la Recherche. Calculations were partially performed at sciCORE (\href{http://scicore.unibas.ch/}{http://scicore.unibas.ch/}) scientific computing center at University of Basel.
\newpage

\appendix

\section{Ball harmonics, rotational correlation, and radial sparsity}
\label{sec:ball_harmonics_and_sparsity}

The ball harmonics provide an orthonormal, rotation-steerable basis for the function space $L^2(\Bc_3)$ that is naturally ordered by Laplacian eigenvalues. This structure makes the representation particularly convenient for rotational alignment via cross-correlation, and it induces a characteristic \emph{radial sparsity} pattern under the eigenvalue-based truncation scheme described by Kileel et al.~\cite{kileel2025fast}.

In this appendix, we give a more precise definition of the ball harmonics, summarize their relevant properties and provide references for the main results.

\subsection{Ball harmonics expansion and eigenvalue-based truncation}
We let $\Bc_3$ denote the unit ball in $\mathbb{R}^3$ and let $u : \Bc_3 \to \mathbb{R}$ be supported on $\Bc_3$. It admits an expansion of the form
\begin{equation}\label{eq:ball-harmonics-expansion}
u(x) = \sum_{(k,\ell,m)\in\mathcal{I}_\Lambda} \hat u_{k,\ell,m}\,\psi_{k,\ell,m}(x),
\end{equation}
where the basis $\{\psi_{k,\ell,m}\}_{k,\ell,m}$ is known as the ball harmonics basis. These basis functions are separable into radial (indexed by $k$) and angular (indexed by $(\ell,m)$)  factors and form an orthonormal basis of $L^2(\Bc_3)$, see \cite{kileel2025fast} for conventions.
The basis functions $\{\psi_{k,\ell,m}\}_{k,\ell,m}$ are Dirichlet Laplacian eigenfunctions on $\Bc_3$ with eigenfrequency \(\lambda_{\ell k}=j_{\ell,k}\), where \(j_{\ell,k}\) denotes the \(k\)-th positive zero of the spherical Bessel function \(j_\ell\). The corresponding Dirichlet Laplacian eigenvalue is \(\lambda_{\ell k}^2\):
\[
-\Delta \psi_{k,\ell,m} = \lambda_{\ell k}^2\,\psi_{k,\ell,m},
\qquad
\psi_{k,\ell,m}=0 \quad \text{on } \partial \Bc_3 .
\]
In practice, we keep only frequencies $\lambda_{\ell k}$ that are below a certain threshold (frequency budget) $\Lambda>0$, and retain the index set
$$
\mathcal{I}_\Lambda := \{(k,\ell,m): \lambda_{\ell k}\le \Lambda,\ \ell\in\mathbb{Z}_{\ge0},\ m\in\{-\ell,\dots,\ell\},\ k\in\mathbb{Z}_{>0}\}.
$$
For a given angular degree $\ell$, define the retained radial set
$$
K_\ell := \{k\in\mathbb{Z}_{>0}:\lambda_{\ell k}\le \Lambda\}.
$$
The function $(\ell,k) \mapsto \lambda_{\ell k}$ is increasing in both $\ell$ and $k$, so the cardinality $|K_\ell|$ decreases monotonically with $\ell$. 
Low-rank approximation can be employed to reduce the number of coefficients in the set $\mathcal{I}_\Lambda$, see Appendix \ref{app:effective_rank}.

\subsection{Computational complexity}
A key practical property is that the truncated expansion in Equation \eqref{eq:ball-harmonics-expansion} can be computed in nearly-linear time (up to polylogarithmic factors) in the number of voxels.

\begin{lemma}[Ball harmonics transform complexity]
\label{lem:complexity_spherical_harmonic}
Given a volume sampled on an $N\times N\times N$ Cartesian grid, the ball harmonics decomposition up to a bandwidth $\Lambda=\Oc (N)$ can be computed with complexity $\Oc (N^3\log^2 N)$ while achieving numerical accuracy $\varepsilon=\Oc(1/N)$.
\end{lemma}

\begin{proof}
This follows from \cite[Theorem~3.1]{kileel2025fast}. Choosing $\varepsilon=1/N$ gives $|\log \varepsilon|=\Theta(\log N)$, so the stated complexity bound reduces to $\Oc (N^3(\log N)^2)$ for $\Lambda= \Oc (N)$.
\end{proof}

\subsection{Parametrization of the rotation group \texorpdfstring{\(\SO(3)\)}{SO(3)}}\label{app:parametrization}
There are many ways to parametrize $\SO(3)$ \cite{stuelpnagel1964parametrization}, including unit quaternions, the
axis--angle (exponential) map, and various Euler-angle conventions.
In the implementation, we parametrize rotations $g\in\Gc$ by ZYZ Euler-angles
following~\cite{kostelec_ffts_2008},
\begin{equation}\label{eq:Euler}
    g \equiv g(\alpha, \beta, \gamma) = r_z(\alpha) r_y(\beta) r_z(\gamma), 
\end{equation}
where
\begin{equation*}
    r_z(\alpha) = \begin{pmatrix}
        \cos\alpha & - \sin \alpha & 0\\
        \sin\alpha & \cos \alpha & 0\\
        0 & 0 & 1
    \end{pmatrix}, \quad r_y(\beta) = \begin{pmatrix}
        \cos\beta & 0 & \sin \beta\\
        0 & 1 & 0 \\
        -\sin\beta & 0 & \cos \beta
    \end{pmatrix},
\end{equation*}
and $(\alpha,\beta,\gamma)\in [0,2\pi)\times[0,\pi]\times[0,2\pi)$. We adopt the
ZYZ convention because it is the natural coordinate system for the harmonic
analysis on $\SO(3)$ underlying our method: the Wigner-$D$ matrices factor as
$D^\ell_{m m'}(\alpha,\beta,\gamma) = e^{-im\alpha}\, d^\ell_{m m'}(\beta)\, e^{-im'\gamma}$,
so that the $\alpha$ and $\gamma$ dependence enters through simple complex
exponentials and only the $\beta$ dependence requires the (real) Wigner
$d$-functions.

\subsection{\texorpdfstring{Wigner--\(D\) matrices}{Wigner-D matrices}}

\label{app:d-matrix}
An orthogonal basis for $L^2(\SO(3))$ is given by the Wigner--\(D\) matrices. 
As a result, rotation of ball harmonic functions can be expressed as a linear combination of other ball harmonic functions and Wigner--\(D\) matrices with angular indices up to the same order. This property is useful for efficiently evaluating the cross-correlation at many rotations and for computing first- and second-order derivatives.
More precisely, for any $g\in\SO(3)$, we have
\begin{equation*}
\psi_{k,\ell,m}(g^{-1}x)
=
\sum_{m'=-\ell}^{\ell}
D^\ell_{m m'}(g)\,
\psi_{k,\ell,m'}(x), \forall x\in\Bc_3.    
\end{equation*}

As defined in Equation~\ref{eq:Euler}, any rotation $g\in\SO(3)$ can be parametrized by ZYZ Euler-angles $(\alpha,\beta,\gamma)\in [0,2\pi)\times[0,\pi]\times[0,2\pi)$.
For each angular degree $\ell\in\mathbb{Z}_{\ge 0}$, the Wigner--\(D\) matrix
\begin{equation*}
D^\ell(g) \in \mathbb{C}^{(2\ell+1)\times(2\ell+1)}, 
\end{equation*}
has entries indexed by order $m,m' \in \{-\ell,\dots,\ell\}$ and is defined by
\begin{equation}\label{eq:wigner-D}
D^\ell_{m m'}(\alpha,\beta,\gamma)
=
e^{- i m \alpha}\,
d^\ell_{m m'}(\beta)\,
e^{- i m' \gamma},
\end{equation}
where $d^\ell_{m m'}$ are the \emph{Wigner--$d$ functions}. The Wigner--\(D\) functions depend only on the polar angle $\beta$ and admit explicit expressions in terms of finite sums \cite{khersonskii_quantum_1988} and can be evaluated efficiently.

The collection $\{D^\ell_{m m'}\}_{\ell,m,m'}$ forms an orthogonal basis of $L^2(\SO(3))$ with respect to the Haar measure $dg$, normalized so that
\begin{equation}\label{eq:wigner-orthogonality}
\int_{\SO(3)} 
D^\ell_{m m'}(g)\,
\overline{D^{\ell'}_{n n'}(g)}\,
dg
=
\frac{1}{2 \ell + 1}\,
\delta_{\ell\ell'}\,\delta_{mn}\,\delta_{m'n'}.
\end{equation}

The Wigner--\(D\) matrices admit closed-form expressions not only for their values but also for their derivatives with respect to the rotation parameters. In particular, derivatives with respect to the Euler-angles $(\alpha,\beta,\gamma)$ can be written analytically in terms of the same Wigner--\(D\) matrices, see \cite{khersonskii_quantum_1988}. Differentiation with respect to $\alpha$ and $\gamma$ acts diagonally,
$$
\partial_\alpha D^\ell_{m m'} = - i m\, D^\ell_{m m'}, 
\qquad
\partial_\gamma D^\ell_{m m'} = - i m'\, D^\ell_{m m'},
$$
while derivatives with respect to \(\beta\) can be expressed using standard recurrence relations for the Wigner--\(D\) functions, equivalently as finite linear combinations of neighboring Wigner--\(D\) functions with known coefficients.

\subsection{Rotational cross-correlation in the ball harmonics basis}
\label{sec:appendix-rotation-cc}

Rotations act only on the angular index via Wigner--\(D\) matrices. Consequently, the rotational cross-correlation of $f,h\in L^2(\Bc_3)$,
\begin{equation*}
\CC(g) = \int_{\Bc_3} f(x)\,\overline{h(g^{-1}x)}\,dx,
\qquad g\in\SO(3),
\end{equation*}
admits a natural Wigner--\(D\) expansion. In a bandlimited setting with angular cutoff $L$, we write
\begin{equation*}
\CC_L(g) = \sum_{\ell=0}^{L}\ \sum_{m=-\ell}^\ell\ \sum_{m'=-\ell}^\ell
\sigma_{\ell m m'}\, D^\ell_{m m'}(g),
\end{equation*}
where $D^\ell_{m m'}(g)$ are the Wigner--\(D\) matrix entries defined in Equation \eqref{eq:wigner-D}.
The coefficients are determined directly by the ball harmonic coefficients of $f$ and $h$, where reindexing and conjugation are absorbed into the coefficients:
\begin{equation}\label{eq:sigmalmm}
\sigma_{\ell m m'} = \sum_{k\in K_\ell} \hat f_{k,\ell,m}\,\overline{\hat h_{k,\ell,m'}}.
\end{equation}
The exact placement of conjugates and the ordering of \(m,m'\) depend on the chosen convention for the Wigner--\(D\) matrices and the group action. The formula above corresponds to the convention in Equation~\eqref{eq:wigner-D}.
For a fixed rotation $g$, a direct evaluation of Equation~\eqref{eq:cc-L} requires
\[
\sum_{\ell=0}^{L}(2\ell+1)^2 = \frac{(L+1)(2L+1)(2L+3)}{3} = \Oc(L^3)
\]
Wigner--\(D\) terms. 
    
\subsection{Low-rank structure induced by radial sparsity}\label{app:analysis_lowrank}

For each fixed $\ell$, define the $(2\ell+1)\times (2\ell+1)$ matrix
$$
A_\ell = \big[\sigma_{\ell m m'}\big]_{m,m'=-\ell}^{\ell}.
$$
Using \eqref{eq:sigmalmm}, $A_\ell$ admits the factorization
\begin{equation}\label{eq:Alpha_l_factorization}
A_\ell = \sum_{k\in K_\ell} a_{\ell k}\, b_{\ell k}^\ast,
\qquad
(a_{\ell k})_m=\hat f_{k,\ell,m},\quad (b_{\ell k})_m=\hat h_{k,\ell,m},
\end{equation}
and therefore
\begin{equation}\label{eq:rank_bound}
\operatorname{rank}(A_\ell)\le |K_\ell|.
\end{equation}
Since eigenvalue-based truncation enforces $|K_\ell|$ to decrease with $\ell$, the coefficient blocks $A_\ell$ of the correlation become progressively low-rank at higher angular degrees. In practice, the effective rank is often even smaller than $|K_\ell|$ due to additional dependencies among the vectors $\{a_{\ell k}\}_{k\in K_\ell}$ and $\{b_{\ell k}\}_{k\in K_\ell}$.

In particular, $\rank(A_\ell)$ does not depend on the rotation of $h$: 
Let $h_{g^*} = h(R_{g^*}^{-1}x)$ be the volume $h$ rotated by $g^*\in \SO(3)$ with $(b_{\ell k})_m^{(g^*)}$ and $A_\ell^{(g^*)}$ the corresponding factorization of the ball harmonics coefficients of $h_{g^*}$. 

If the rotated coefficients satisfy
\[
b_{\ell k}^{(g^\star)}=D^\ell(g^\star)b_{\ell k},
\]
then
\[
A_\ell^{(g^\star)}
=\sum_{k\in K_\ell} a_{\ell k}\bigl(b_{\ell k}^{(g^\star)}\bigr)^*
=\sum_{k\in K_\ell} a_{\ell k} b_{\ell k}^* D^\ell(g^\star)^*
=A_\ell D^\ell(g^\star)^*.
\]
Since \(D^\ell(g^\star)\) is unitary, multiplication by
\(D^\ell(g^\star)^*\) preserves rank.
\subsection{Bernstein-type inequality}

Bandlimited functions on $\SO(3)$ satisfy Bernstein-type derivative bounds.
\begin{lemma}[Bernstein inequality on $\SO(3)$]
\label{lem:bernstein_SO3}
Let $\Pi_L \eqdef \operatorname{span}\{D^\ell_{m m'}:\ 0\le \ell\le L\}$ be the space of functions band-limited to degree $L$. There exist constants $B_1,B_2>0$, depending only on the metric normalization, such that every $F\in \Pi_L$ satisfies
\[
  \|\nabla F\|_{L^\infty(\SO(3))} \le B_1(1+L)\,\|F\|_{L^\infty(\SO(3))},
  \qquad
  \|\nabla^2 F\|_{L^\infty_{\mathrm{op}}(\SO(3))} \le B_2(1+L)^2\,\|F\|_{L^\infty(\SO(3))}.
\]
\end{lemma}

\begin{proof}
Let $X_1,X_2,X_3$ be an orthonormal basis of left-invariant vector fields for the fixed metric, and let $\mathcal L \eqdef -(X_1^2+X_2^2+X_3^2)$ be the associated Laplace--Beltrami operator. Each Wigner function $D^\ell_{m m'}$ is an eigenfunction of $\mathcal L$ with eigenvalue $\kappa\,\ell(\ell+1)$, where $\kappa>0$ depends only on the metric normalization; hence every $F\in\Pi_L$ has $\mathcal L$-spectrum contained in $[0,\,\kappa L(L+1)]$.

By Pesenson's Bernstein inequality for band-limited functions on compact homogeneous manifolds \cite{GellerPesenson2011,Pesenson2008BernsteinNikolskii}, applied at $p=\infty$, this spectral bound yields
\[
  \max_i \|X_iF\|_\infty \le C_1\,\sqrt{L(L+1)}\,\|F\|_\infty,
  \qquad
  \max_{i,j} \|X_iX_jF\|_\infty \le C_2\,L(L+1)\,\|F\|_\infty,
\]
with $C_1,C_2$ depending only on $\kappa$. In particular, passing to the sup-norm incurs no dimensional loss, so the scaling is the same as in $L^2$.

The gradient bound follows from $|\nabla F|^2=\sum_i |X_iF|^2$ and $\sqrt{L(L+1)}\le 1+L$. For the Hessian, $\nabla^2 F(X_i,X_j)=X_iX_jF-(\nabla_{X_i}X_j)F$; bi-invariance gives $\nabla_{X_i}X_j=\tfrac12[X_i,X_j]$, a fixed linear combination of $X_1,X_2,X_3$, so each entry satisfies
\[
  |\nabla^2 F(X_i,X_j)|
  \le \max_{p,q}\|X_pX_qF\|_\infty + c\,\max_p\|X_pF\|_\infty
  \le C\,L(L+1)\,\|F\|_\infty,
\]
where $c$ is an absolute constant arising from the brackets $[X_i,X_j]$ and $C$ depends only on $\kappa$. Bounding the operator norm of this matrix by its entries and using $L(L+1)\le(1+L)^2$ completes the proof.
\end{proof}
\section{Additional Experiments} 

\subsection{Empirical effective rank of Wigner blocks versus the theoretical bound}
\label{app:effective_rank}

In Appendix~\ref{app:analysis_lowrank} we remark that, for each angular degree $\ell$, the Wigner block
$A_\ell\in\C^{(2\ell+1)\times(2\ell+1)}$ of the rotational correlation admits a factorization such that $\rank(A_\ell)\le |K_\ell|$, where $K_\ell$ is the set of radial indices retained at degree $\ell$ by the eigenvalue cutoff $\Lambda$.
This provides a \emph{theoretical} (worst-case) rank upper bound determined purely by the truncation scheme.
In practice, however, we observe substantially lower rank behavior.

For selected degrees $\ell\in\{20,30,50,80\}$, see Figure~\ref{fig:block_sv}, the singular values exhibit a pronounced spectral gap: only a small number of modes carry significant mass, while the remaining singular values rapidly decay to a near-zero floor.
Motivated by this behavior, we define a gap-based effective rank by
$$
r_{\mathrm{eff}}(\ell)
\eqdef \min\{k: \frac{s_{k+1} (A_\ell)}{s_k(A_\ell)} < \rho\},
$$ 
where $s_1(A_\ell)\ge s_2(A_\ell)\ge\cdots$ are the singular values of $A_\ell$ and $\rho$ is a fixed numerical tolerance.

Figure~\ref{fig:block:rk} shows $r_{\mathrm{eff}}(\ell)$ alongside the upper-bound $|K_\ell|$ across degrees for $\rho = 0.2$.
Across degrees, $r_{\mathrm{eff}}(\ell)$ stays well below $|K_\ell|$ with high-$\ell$ blocks approaching rank one. This indicates that the correlation is effectively supported on a low-dimensional subspace relative to the ambient  $(2\ell+1)\times(2\ell+1)$ block size.

\def\a{0.7}
\def\mz{1.5pt}
\begin{figure}
	\centering
    \begin{subfigure}[t]{0.48\textwidth}
	\begin{tikzpicture}     
		\begin{axis}[
            ymode=log,   
			width=0.99\linewidth, % Scale the plot to \linewidth
			grid=major, % Display a grid
			grid style={dashed,gray!30}, % Set the style
			xlabel= {\scriptsize Singular value index},
			ylabel={\scriptsize Singular value amplitude (log)},
			xtick={0,50,100,150,200,250,300,350},
			xmin = 0, xmax = 162,
			axis x line*=bottom,
			axis y line*=left,
			legend style={at={(1.,0.9)}, legend cell align=left, align=left, draw=none,font=\scriptsize},
			xticklabel style={
				/pgf/number format/fixed,
				/pgf/number format/precision=0,
				/pgf/number format/fixed zerofill
			},
			scaled x ticks=false,
            ticklabel style={font=\scriptsize}
			]
			\addplot[only marks, color=blue, mark=*, mark size = \mz, opacity=\a] table [x expr=\coordindex, y=s_20, col sep=comma] {img/Experiment_rank/singular_values_set.csv};
			\addlegendentry{$l=20$}
			\addplot[only marks, color=orange, mark=*, mark size = \mz, opacity=\a] table [x expr=\coordindex, y=s_30, col sep=comma] {img/Experiment_rank/singular_values_set.csv};
			\addlegendentry{$l=30$}
			\addplot[only marks, color=olive, mark=*, mark size = \mz, opacity=\a] table [x expr=\coordindex, y=s_50, col sep=comma] {img/Experiment_rank/singular_values_set.csv};
			\addlegendentry{$l=50$}
			\addplot[only marks, color=violet, mark=*, mark size = \mz, opacity=\a] table [x expr=\coordindex, y=s_80, col sep=comma] {img/Experiment_rank/singular_values_set.csv};
			\addlegendentry{$l=80$}
		\end{axis}   
	\end{tikzpicture}	
	\caption{Singular values of $A_\ell$ for selected degrees.} 
    \label{fig:block_sv}
    \end{subfigure}
    \hfill
    \begin{subfigure}[t]{0.48\textwidth}
	\begin{tikzpicture}     
		\begin{axis}[
            % ymode=log,   
			width=0.99\linewidth, % Scale the plot to \linewidth
			grid=major, % Display a grid
			grid style={dashed,gray!30}, % Set the style
			xlabel= {\scriptsize Degree $\ell$},
			ylabel={\scriptsize Rank},
			xtick={0,50,100,150,200,250,300,350},
			xmin = 0, xmax = 86,
			ymin = 0,
			axis x line*=bottom,
			axis y line*=left,
			legend style={at={(1.,0.9)}, legend cell align=left, align=left, draw=none,font=\scriptsize},
			xticklabel style={
				/pgf/number format/fixed,
				/pgf/number format/precision=0,
				/pgf/number format/fixed zerofill
			},
			scaled x ticks=false,
            ticklabel style={font=\scriptsize}
			]
			\addplot[only marks, color=blue, mark=*, mark size = \mz, opacity=\a] table [x expr=\coordindex, y=effective_ranks, col sep=comma] {img/Experiment_rank/effective_ranks.csv};
			\addlegendentry{Effective rank}
			\addplot[only marks, color=orange, mark=*, mark size = \mz, opacity=\a] table [x expr=\coordindex, y=theoretical_sizes, col sep=comma] {img/Experiment_rank/theoretical_sizes.csv};
			\addlegendentry{Upper-bound $|K_\ell|$}
		\end{axis}   
	\end{tikzpicture}	
	\caption{Effective rank vs.\ theoretical upper bound $|K_\ell|$.}  \label{fig:block:rk} 
    \end{subfigure}
    \caption{Spectral properties of the Wigner blocks $A_\ell$.
  In~(a), the singular values exhibit a sharp spectral gap at
  each degree, with only a few modes carrying significant energy.
  In~(b), the effective rank (defined by the ratio criterion
  $s_{k+1}/s_k < \rho$ with $\rho = 0.2$) stays well below
  the theoretical bound~$|K_\ell|$ across all degrees, with
  high-$\ell$ blocks approaching rank one. This low-rank structure
  can be exploited to accelerate the evaluation of~$\CC_L$.}
\end{figure}
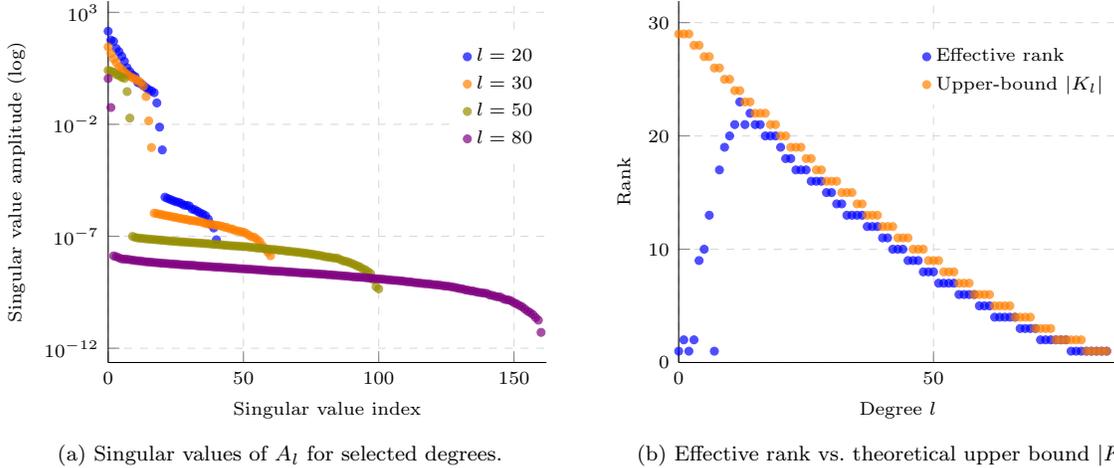

\subsection{Ablation study for other noise regimes}\label{sec:other_noise}
We follow the experimental setting presented in Section \ref{sec:exp-synthetic} and display the results for the extremely noisy setting, SNR$=-20\,\mathrm{dB}$, and the easier setting, SNR$=20\,\mathrm{dB}$.

A key hyperparameter of \method{} and SOFFT is the maximum bandwidth $L_{\mathrm{max}}$ used to evaluate the truncated correlation in Equation~\eqref{eq:cc-L}. This parameter controls the highest angular frequencies included in the alignment score and trades high-resolution information for computational time.

With \method{}, the experiments indicate that accurate alignment can already be achieved at relatively low bandwidth when the signal-to-noise ratio is moderate or high. As shown in Figures~\ref{fig:accuracy-v-bandlimit-snr0db}, the alignment accuracy remains high across the entire range of tested $L_{\mathrm{max}}$, suggesting that low-frequency components alone are sufficient to reliably localize the correct rotation in these regimes. 

At very low SNR, e.g. SNR $=-20$\,dB, Figure~\ref{fig:accuracy-v-bandlimit-snr-20}, the overall accuracy of both methods degrades, reflecting the increased difficulty of the alignment problem. Nevertheless, both methods continue to benefit from increasing the bandwidth. For \method{}, higher values of $L_{\mathrm{max}}$ lead to a gradual improvement in accuracy even in this strongly noisy setting, although the gains become progressively smaller as additional high-frequency components are increasingly dominated by noise. A similar trend is observed for SOFFT.

Overall, these results show that, as expected, increasing the bandwidth is beneficial for both methods across all noise levels. At the same time, \method{} exhibits strong robustness to bandwidth truncation at moderate and high SNR, achieving high accuracy even at low $L_{\mathrm{max}}$, while still being able to exploit additional frequency information when available.

Figure~\ref{fig:accuracy-v-time20db} and Figure~\ref{fig:accuracy-v-time-20db} present the execution times for \method{} and SOFFT with varying oversampling factor $K$ at SNR 20dB and -20dB respectively. In general, execution time for both \method{} and SOFFT does not depend on the noise level.

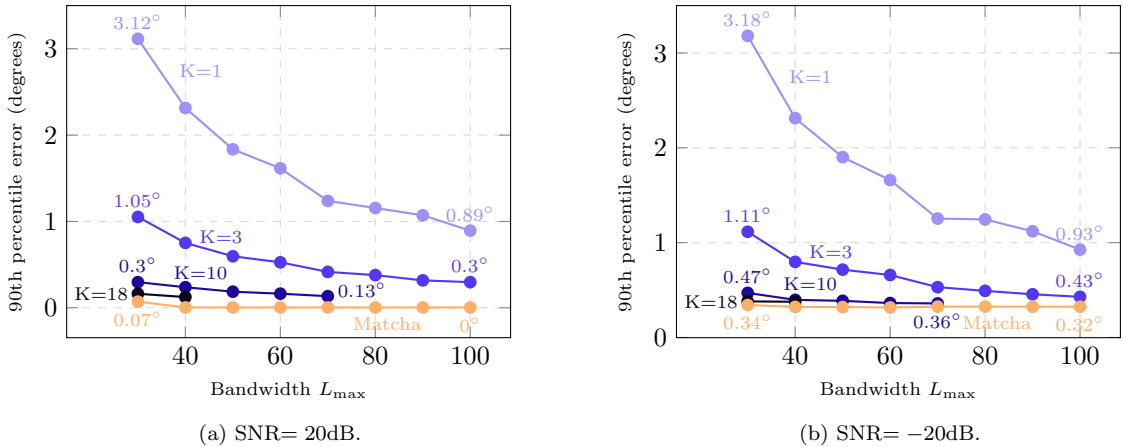
\begin{figure}[htbp]
    \centering
    \begin{subfigure}[b]{0.48\textwidth}
            \begin{tikzpicture}
         \begin{axis}[
            % ymode=log,
            xlabel={\scriptsize Bandwidth $L_\mathrm{max}$},
            ylabel={\scriptsize 90th percentile error (degrees)},
            ymax=3.5,
            width=\textwidth,
			grid=major, % Display a grid
			grid style={dashed,gray!30}, % Set the style
            height=6cm,
            xmin = 15,
            xtick={40,60,80,100},
        ]
            \newcommand{\snr}{snr20}
            % This should plot if the file is readable and has x,y columns:
            \PlotCurveWithLabels{plots/plot1_newlayout/ov1_\snr.csv}{cOv1}{\hspace{0.7cm} K=1}{solid, thick}{above}{0.1}{above}{0pt}{0pt}

            \PlotCurveWithLabels{plots/plot1_newlayout/ov3_\snr.csv}{cOv3}{K=3}{solid, thick}{above}{0.25}{above}{0pt}{0pt}

            \PlotCurveWithLabels{plots/plot1_newlayout/ov10_\snr.csv}{cOv10}{K=10}{solid, thick}{above}{0.33}{above}{13pt}{-4pt}
            \PlotCurveWithLabelsSecond{plots/plot1_newlayout/ov18_\snr.csv}{cOv18}{K=18}{solid, thick}{}{0.0}

            \PlotCurveWithLabels{plots/plot1_newlayout/ours_\snr.csv}{cOurs}{\method{}}{solid, thick}{below}{0.75}{below}{0pt}{0pt}
        \end{axis}
    \end{tikzpicture}
    \caption{SNR$=20$dB.}
    \label{fig:accuracy-v-bandlimit-snr20}
    \end{subfigure}
    \hfill
        \begin{subfigure}[b]{0.48\textwidth}
                    \begin{tikzpicture}
         \begin{axis}[
            xlabel={\scriptsize Bandwidth $L_\mathrm{max}$},
            ylabel={\scriptsize 90th percentile error (degrees)},
            ymax=3.5,
            width=\textwidth,
			grid=major, % Display a grid
			grid style={dashed,gray!30}, % Set the style
            height=6cm,
            xmin = 15,
            xtick={40,60,80,100},
        ]
            \newcommand{\snr}{snr-20}
            \PlotCurveWithLabels{plots/plot1_newlayout/ov1_\snr.csv}{cOv1}{\hspace{0.7cm} K=1}{solid, thick}{above}{0.1}{above}{0pt}{0pt}
            \PlotCurveWithLabels{plots/plot1_newlayout/ov3_\snr.csv}{cOv3}{K=3}{solid, thick}{above}{0.25}{above}{0pt}{0pt}
            \PlotCurveWithLabels{plots/plot1_newlayout/ov10_\snr.csv}{cOv10}{K=10}{solid, thick}{above}{0.33}{above}{0pt}{-13pt}
            \PlotCurveWithLabelsSecond{plots/plot1_newlayout/ov18_\snr.csv}{cOv18}{K=18}{solid, thick}{}{0.0}
            \PlotCurveWithLabels{plots/plot1_newlayout/ours_\snr.csv}{cOurs}{\method{}}{solid, thick}{below}{0.75}{below}{0pt}{0pt}
        \end{axis}
    \end{tikzpicture}
    \caption{SNR$=-20$dB.}
    \label{fig:accuracy-v-bandlimit-snr-20}
    \end{subfigure}
    \caption{Alignment accuracy of \method{} and SOFFT with varying oversampling factor $K$ for different $L_\mathrm{max}$ for different SNR.}
\end{figure}

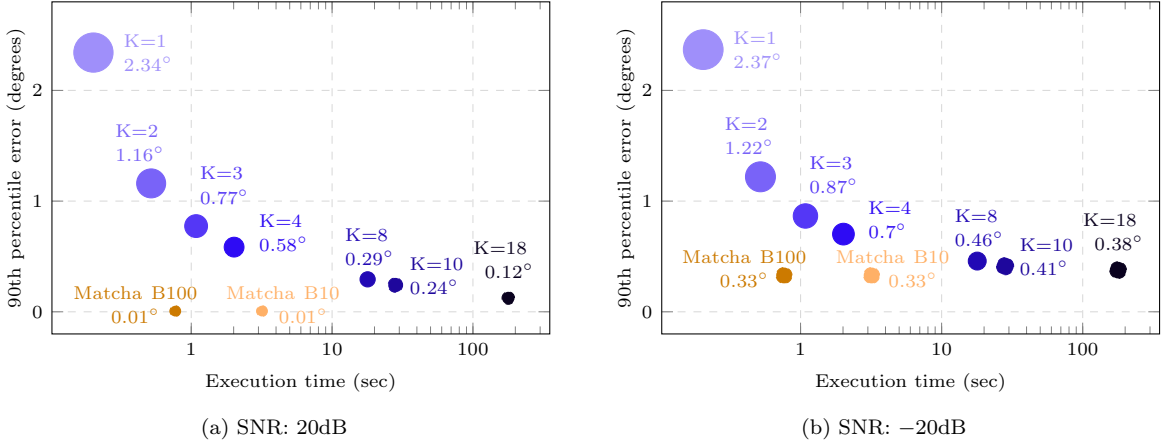
\begin{figure}
    \centering
     \begin{subfigure}[b]{0.48\textwidth}
    \centering
    \begin{tikzpicture}
        \begin{axis}[
            xmode=log,
            xlabel={\scriptsize Execution time (sec)},
            ylabel={\scriptsize90th percentile error (degrees)},
            ymin=-0.2,
            ymax=2.8,
            xtick={0.001, 0.01,0.1,1},
            xticklabels={1,10,100,1000},
            width=8.2cm,
            height=6cm,
            ticklabel style={font=\scriptsize},
            ylabel style={yshift=-0.15cm},
			grid=major, % Display a grid
			grid style={dashed,gray!30}, % Set the style
        ]
        
            \newcommand{\snr}{snr20}
            \newcommand{\pathPlotDots}{plots/plot2_newlayout/}
            \DTLifdbexists{mydata}{\DTLcleardb{mydata}\DTLdeletedb{mydata}}{}
            \DTLloaddb{mydata}{\pathPlotDots/sofft_K1_\snr.csv}
            \DTLgetvalue{\alphavalA}{mydata}{1}{2}
            \pgfmathparse{round(\alphavalA*100)/100}
            \edef\alphavalroundedA{\pgfmathresult}
            \pgfmathparse{(\alphavalroundedA*1.5+0.2)}
            \edef\alphascaleA{\pgfmathresult}
            \addplot+[mark=*, color=cOv1, mark options={scale=\alphascaleA}]
              table[col sep=comma, x=x, y=y] {\pathPlotDots/sofft_K1_\snr.csv} node[pos=0,right, text width=5cm] {\scriptsize{\quad K=1 \\\vspace{-0.1cm}\quad$\pgfmathprintnumber[fixed,precision=2]{\alphavalA}^\circ$}};
            \DTLgetvalue{\xA}{mydata}{1}{1}
            \DTLgetvalue{\yA}{mydata}{1}{2}
           
            \DTLcleardb{mydata}
            \DTLdeletedb{mydata}
            
            \DTLloaddb{mydata}{\pathPlotDots/sofft_K2_\snr.csv}
            \DTLgetvalue{\alphavalB}{mydata}{1}{2}
            \pgfmathparse{round(\alphavalB*100)/100.}
            \edef\alphavalroundedB{\pgfmathresult}
            \pgfmathparse{(\alphavalroundedB*1.5 + 1)}
            \edef\alphascaleB{\pgfmathresult}
            \addplot+[mark=*, color=cOv2, mark options={scale=\alphascaleB}]
              table[col sep=comma, x=x, y=y] {\pathPlotDots/sofft_K2_\snr.csv} node[pos=0,text width=5cm, anchor=south west, xshift= -25pt, yshift=5pt] {\scriptsize{\quad K=2 \\\vspace{-0.1cm}\quad$\pgfmathprintnumber[fixed,precision=2]{\alphavalB}^\circ$}};
            \DTLcleardb{mydata}
            \DTLdeletedb{mydata}
            
            \DTLloaddb{mydata}{\pathPlotDots/sofft_K3_\snr.csv}
            \DTLgetvalue{\alphavalC}{mydata}{1}{2}
            \pgfmathparse{round(\alphavalC*100)/100}
            \edef\alphavalroundedC{\pgfmathresult}
            \pgfmathparse{(\alphavalroundedC*1.5 + 1)}
            \edef\alphascaleC{\pgfmathresult}
            \addplot+[mark=*, color=cOv3, mark options={scale=\alphascaleC}]
              table[col sep=comma, x=x, y=y] {\pathPlotDots/sofft_K3_\snr.csv} node[pos=0,right, text width=5cm, anchor=south west, yshift=5pt, xshift= -10pt] {\scriptsize{\quad K=3 \\\vspace{-0.1cm}\quad$\pgfmathprintnumber[fixed,precision=2]{\alphavalC}^\circ$}};
            \DTLcleardb{mydata}
            \DTLdeletedb{mydata}
            
            \DTLloaddb{mydata}{\pathPlotDots/sofft_K4_\snr.csv}
            \DTLgetvalue{\alphavalD}{mydata}{1}{2}
            \pgfmathparse{round(\alphavalD*100)/100}
            \edef\alphavalroundedD{\pgfmathresult}
            \pgfmathparse{(\alphavalroundedD*1.5+ 1)}
            \edef\alphascaleD{\pgfmathresult}
            \addplot+[mark=*, color=cOv5, mark options={scale=\alphascaleD}]
              table[col sep=comma, x=x, y=y] {\pathPlotDots/sofft_K4_\snr.csv} node[pos=0,right, text width=5cm, anchor=south west, yshift=-5pt, xshift=-2pt]{\scriptsize{\quad K=4 \\\vspace{-0.1cm}\quad$\pgfmathprintnumber[fixed,precision=2]{\alphavalD}^\circ$}};
            \DTLcleardb{mydata}
            \DTLdeletedb{mydata}
            
            \DTLloaddb{mydata}{\pathPlotDots/sofft_K8_\snr.csv}
            \DTLgetvalue{\alphavalE}{mydata}{1}{2}
            \pgfmathparse{round(\alphavalE*100)/100}
            \edef\alphavalroundedE{\pgfmathresult}
            \pgfmathparse{(\alphavalroundedE*1.5 + 1)}
            \edef\alphascaleE{\pgfmathresult}
            \addplot+[mark=*, color=cOv8, mark options={scale=\alphascaleE}]
              table[col sep=comma, x=x, y=y] {\pathPlotDots/sofft_K8_\snr.csv} node[pos=0,right, text width=5cm,anchor=south west, yshift=2pt, xshift=-20pt]{\scriptsize{\quad K=8\\\vspace{-0.1cm}\quad $\pgfmathprintnumber[fixed,precision=2]{\alphavalE}^\circ$}};
            \DTLcleardb{mydata}
            \DTLdeletedb{mydata}
            
            \DTLloaddb{mydata}{\pathPlotDots/sofft_K10_\snr.csv}
            \DTLgetvalue{\alphavalF}{mydata}{1}{2}
            \pgfmathparse{round(\alphavalF*100)/100}
            \edef\alphavalroundedF{\pgfmathresult}
            \pgfmathparse{(\alphavalroundedF*1.5 + 1)}
            \edef\alphascaleF{\pgfmathresult}
            \addplot+[mark=*, color=cOv10, mark options={scale=\alphascaleF}]
              table[col sep=comma, x=x, y=y] {\pathPlotDots/sofft_K10_\snr.csv} node[right, text width=5cm,anchor=south west, yshift=-7pt, xshift=-6pt]{\scriptsize{\quad K=10\\\vspace{-0.1cm}\quad$\pgfmathprintnumber[fixed,precision=2]{\alphavalF}^\circ$}};
            \DTLcleardb{mydata}
            \DTLdeletedb{mydata}
            
            \DTLloaddb{mydata}{\pathPlotDots/sofft_K18_\snr.csv}
            \DTLgetvalue{\alphavalL}{mydata}{1}{2}
            \pgfmathparse{round(\alphavalL*100)/100}
            \edef\alphavalroundedL{\pgfmathresult}
            \pgfmathparse{(\alphavalroundedL*1.5+ 1)}
            \edef\alphascaleL{\pgfmathresult}
            \addplot+[mark=*, color=cOv18, mark options={scale=\alphascaleL}]
              table[col sep=comma, x=x, y=y] {\pathPlotDots/sofft_K18_\snr.csv} node[pos=0,below, text width=5cm,anchor=north, xshift=50pt, yshift=25pt] {\scriptsize{\quad K=18\\\vspace{-0.1cm}\hspace{0.15cm}\quad$\pgfmathprintnumber[fixed,precision=2]{\alphavalL}^\circ$}};
            \DTLcleardb{mydata}
            \DTLdeletedb{mydata}
            
            \DTLloaddb{mydata}{\pathPlotDots/matcha_bs10_steps1_cand10_\snr.csv}
            \DTLgetvalue{\alphavalH}{mydata}{1}{2}
            \pgfmathparse{round(\alphavalH*100)/100}
            \edef\alphavalroundedH{\pgfmathresult}
            \pgfmathparse{(\alphavalroundedH*1.5 + 1)}
            \edef\alphascaleH{\pgfmathresult}
            \addplot+[mark=*, color=cOurs10, mark options={scale=\alphascaleH}]
              table[col sep=comma, x=x, y=y] {\pathPlotDots/matcha_bs10_steps1_cand10_\snr.csv} node[pos=0,right, text width=5cm,anchor=south west, yshift=-8pt, xshift=-25pt] {\scriptsize{\quad \method{} B10\\\vspace{-0.1cm}\hspace{0.7cm}\quad $\pgfmathprintnumber[fixed,precision=2]{\alphavalH}^\circ$}};
            \DTLcleardb{mydata}
            \DTLdeletedb{mydata}
            
            \DTLloaddb{mydata}{\pathPlotDots/matcha_bs100_steps1_cand10_\snr.csv}
            \DTLgetvalue{\alphavalI}{mydata}{1}{2}
            \pgfmathparse{round(\alphavalI*100)/100}
            \edef\alphavalroundedI{\pgfmathresult}
            \pgfmathparse{(\alphavalroundedI*1.5 + 1)}
            \edef\alphascaleI{\pgfmathresult}
            \addplot+[mark=*, color=cOurs100, mark options={scale=\alphascaleI}]
              table[col sep=comma, x=x, y=y] {\pathPlotDots/matcha_bs100_steps1_cand10_\snr.csv} node[pos=0, above, text width=5cm,anchor=south west, yshift=-8pt, xshift=-50pt] {\scriptsize{\quad \method{} B100\\\vspace{-0.1cm}\hspace{0.5cm}\quad $\pgfmathprintnumber[fixed,precision=2]{\alphavalI}^\circ$}};
            \DTLcleardb{mydata}
            \DTLdeletedb{mydata}
        \end{axis}
    \end{tikzpicture}
    \caption{SNR: $20$dB} \label{fig:accuracy-v-time20db}
    \end{subfigure} 
    \hfill
             \begin{subfigure}[b]{0.48\textwidth}
        \centering
        \begin{tikzpicture}
            \begin{axis}[
                xmode=log,
                xlabel={\scriptsize Execution time (sec)},
                ylabel={\scriptsize90th percentile error (degrees)},
                ymin=-0.2,
                ymax=2.8,
                xtick={0.001, 0.01,0.1,1},
                xticklabels={1,10,100,1000},
                width=8.2cm,
                height=6cm,
                ticklabel style={font=\scriptsize},
                ylabel style={yshift=-0.15cm},
			grid=major, % Display a grid
			grid style={dashed,gray!30}, % Set the style
            ]
            
                \newcommand{\snr}{snr-20}
                \newcommand{\pathPlotDots}{plots/plot2_newlayout/}
                \DTLifdbexists{mydata}{\DTLcleardb{mydata}\DTLdeletedb{mydata}}{}
                \DTLloaddb{mydata}{\pathPlotDots/sofft_K1_\snr.csv}
                \DTLgetvalue{\alphavalA}{mydata}{1}{2}
                \pgfmathparse{round(\alphavalA*100)/100}
                \edef\alphavalroundedA{\pgfmathresult}
                \pgfmathparse{(\alphavalroundedA*1.5+0.2)}
                \edef\alphascaleA{\pgfmathresult}
                \addplot+[mark=*, color=cOv1, mark options={scale=\alphascaleA}]
                  table[col sep=comma, x=x, y=y] {\pathPlotDots/sofft_K1_\snr.csv} node[pos=0,right, text width=5cm] {\scriptsize{\quad K=1 \\\vspace{-0.1cm}\quad$\pgfmathprintnumber[fixed,precision=2]{\alphavalA}^\circ$}};
                \DTLgetvalue{\xA}{mydata}{1}{1}
                \DTLgetvalue{\yA}{mydata}{1}{2}
               
                \DTLcleardb{mydata}
                \DTLdeletedb{mydata}
                
                \DTLloaddb{mydata}{\pathPlotDots/sofft_K2_\snr.csv}
                \DTLgetvalue{\alphavalB}{mydata}{1}{2}
                \pgfmathparse{round(\alphavalB*100)/100.}
                \edef\alphavalroundedB{\pgfmathresult}
                \pgfmathparse{(\alphavalroundedB*1.5 + 1)}
                \edef\alphascaleB{\pgfmathresult}
                \addplot+[mark=*, color=cOv2, mark options={scale=\alphascaleB}]
                  table[col sep=comma, x=x, y=y] {\pathPlotDots/sofft_K2_\snr.csv} node[pos=0,text width=5cm, anchor=south west, xshift= -25pt, yshift=5pt] {\scriptsize{\quad K=2 \\\vspace{-0.1cm}\quad$\pgfmathprintnumber[fixed,precision=2]{\alphavalB}^\circ$}};
                \DTLcleardb{mydata}
                \DTLdeletedb{mydata}
                
                \DTLloaddb{mydata}{\pathPlotDots/sofft_K3_\snr.csv}
                \DTLgetvalue{\alphavalC}{mydata}{1}{2}
                \pgfmathparse{round(\alphavalC*100)/100}
                \edef\alphavalroundedC{\pgfmathresult}
                \pgfmathparse{(\alphavalroundedC*1.5 + 1)}
                \edef\alphascaleC{\pgfmathresult}
                \addplot+[mark=*, color=cOv3, mark options={scale=\alphascaleC}]
                  table[col sep=comma, x=x, y=y] {\pathPlotDots/sofft_K3_\snr.csv} node[pos=0,right, text width=5cm, anchor=south west, yshift=5pt, xshift= -10pt] {\scriptsize{\quad K=3\\\vspace{-0.1cm}\quad$\pgfmathprintnumber[fixed,precision=2]{\alphavalC}^\circ$}};
                \DTLcleardb{mydata}
                \DTLdeletedb{mydata}
                
                \DTLloaddb{mydata}{\pathPlotDots/sofft_K4_\snr.csv}
                \DTLgetvalue{\alphavalD}{mydata}{1}{2}
                \pgfmathparse{round(\alphavalD*100)/100}
                \edef\alphavalroundedD{\pgfmathresult}
                \pgfmathparse{(\alphavalroundedD*1.5+ 1)}
                \edef\alphascaleD{\pgfmathresult}
                \addplot+[mark=*, color=cOv5, mark options={scale=\alphascaleD}]
                  table[col sep=comma, x=x, y=y] {\pathPlotDots/sofft_K4_\snr.csv} node[pos=0,right, text width=5cm, anchor=south west, yshift=-5pt, xshift=-2pt]{\scriptsize{\quad K=4\\\vspace{-0.1cm}\quad$\pgfmathprintnumber[fixed,precision=2]{\alphavalD}^\circ$}};
                \DTLcleardb{mydata}
                \DTLdeletedb{mydata}
                
                \DTLloaddb{mydata}{\pathPlotDots/sofft_K8_\snr.csv}
                \DTLgetvalue{\alphavalE}{mydata}{1}{2}
                \pgfmathparse{round(\alphavalE*100)/100}
                \edef\alphavalroundedE{\pgfmathresult}
                \pgfmathparse{(\alphavalroundedE*1.5 + 1)}
                \edef\alphascaleE{\pgfmathresult}
                \addplot+[mark=*, color=cOv8, mark options={scale=\alphascaleE}]
                  table[col sep=comma, x=x, y=y] {\pathPlotDots/sofft_K8_\snr.csv} node[pos=0,right, text width=5cm,anchor=south west, yshift=2pt, xshift=-20pt]{\scriptsize{\quad K=8\\\vspace{-0.1cm}\quad $\pgfmathprintnumber[fixed,precision=2]{\alphavalE}^\circ$}};
                \DTLcleardb{mydata}
                \DTLdeletedb{mydata}
                
                \DTLloaddb{mydata}{\pathPlotDots/sofft_K10_\snr.csv}
                \DTLgetvalue{\alphavalF}{mydata}{1}{2}
                \pgfmathparse{round(\alphavalF*100)/100}
                \edef\alphavalroundedF{\pgfmathresult}
                \pgfmathparse{(\alphavalroundedF*1.5 + 1)}
                \edef\alphascaleF{\pgfmathresult}
                \addplot+[mark=*, color=cOv10, mark options={scale=\alphascaleF}]
                  table[col sep=comma, x=x, y=y] {\pathPlotDots/sofft_K10_\snr.csv} node[right, text width=5cm,anchor=south west, yshift=-7pt, xshift=-6pt]{\scriptsize{\quad K=10\\\vspace{-0.1cm}\quad$\pgfmathprintnumber[fixed,precision=2]{\alphavalF}^\circ$}};
                \DTLcleardb{mydata}
                \DTLdeletedb{mydata}
                
                \DTLloaddb{mydata}{\pathPlotDots/sofft_K18_\snr.csv}
                \DTLgetvalue{\alphavalL}{mydata}{1}{2}
                \pgfmathparse{round(\alphavalL*100)/100}
                \edef\alphavalroundedL{\pgfmathresult}
                \pgfmathparse{(\alphavalroundedL*1.5+ 1)}
                \edef\alphascaleL{\pgfmathresult}
                \addplot+[mark=*, color=cOv18, mark options={scale=\alphascaleL}]
                  table[col sep=comma, x=x, y=y] {\pathPlotDots/sofft_K18_\snr.csv} node[pos=0,below, text width=5cm,anchor=north, xshift=50pt, yshift=25pt] {\scriptsize{\quad K=18\\\vspace{-0.1cm}\hspace{0.15cm}\quad$\pgfmathprintnumber[fixed,precision=2]{\alphavalL}^\circ$}};
                \DTLcleardb{mydata}
                \DTLdeletedb{mydata}
                
                \DTLloaddb{mydata}{\pathPlotDots/matcha_bs10_steps1_cand10_\snr.csv}
                \DTLgetvalue{\alphavalH}{mydata}{1}{2}
                \pgfmathparse{round(\alphavalH*100)/100}
                \edef\alphavalroundedH{\pgfmathresult}
                \pgfmathparse{(\alphavalroundedH*1.5 + 1)}
                \edef\alphascaleH{\pgfmathresult}
                \addplot+[mark=*, color=cOurs10, mark options={scale=\alphascaleH}]
                  table[col sep=comma, x=x, y=y] {\pathPlotDots/matcha_bs10_steps1_cand10_\snr.csv} node[pos=0,right, text width=5cm,anchor=south west, yshift=-8pt, xshift=-25pt] {\scriptsize{\quad \method{} B10\\\vspace{-0.1cm}\hspace{0.7cm}\quad $\pgfmathprintnumber[fixed,precision=2]{\alphavalH}^\circ$}};
                \DTLcleardb{mydata}
                \DTLdeletedb{mydata}
                
                \DTLloaddb{mydata}{\pathPlotDots/matcha_bs100_steps1_cand10_\snr.csv}
                \DTLgetvalue{\alphavalI}{mydata}{1}{2}
                \pgfmathparse{round(\alphavalI*100)/100}
                \edef\alphavalroundedI{\pgfmathresult}
                \pgfmathparse{(\alphavalroundedI*1.5 + 1)}
                \edef\alphascaleI{\pgfmathresult}
                \addplot+[mark=*, color=cOurs100, mark options={scale=\alphascaleI}]
                  table[col sep=comma, x=x, y=y] {\pathPlotDots/matcha_bs100_steps1_cand10_\snr.csv} node[pos=0, above, text width=5cm,anchor=south west, yshift=-8pt, xshift=-50pt] {\scriptsize{\quad \method{} B100\\\vspace{-0.1cm}\hspace{0.5cm}\quad $\pgfmathprintnumber[fixed,precision=2]{\alphavalI}^\circ$}};
                \DTLcleardb{mydata}
                \DTLdeletedb{mydata}
            \end{axis}
        \end{tikzpicture}
        \caption{SNR: $-20$dB} \label{fig:accuracy-v-time-20db}
    \end{subfigure} 
    \caption{Alignment accuracy of \method{} and SOFFT versus execution time for $L_{\mathrm{max}} = 40$ and varying SNR. Large dot sizes correspond to larger alignment error. \method{} is one order of magnitude more precise and several orders of magnitude faster than exhaustive search SOFFT.}

\end{figure}

\section{Alternating search over rotations and translations}
\label{sec:appendix-6d-search}

In subtomogram alignment, the relative pose between two volumes is typically described by a rotation and a translation.
Given two volumetric densities $f,h\in L^2(\Bc_3)$, we consider the \emph{6-dimensional alignment} problem
\begin{equation}
\label{eq:6d_problem}
(\hat\tau,\hat g)\in\argmax_{\tau\in\mathbb{R}^3,\ g\in\SO(3)}\ \CC(\tau,g),
\qquad 
\CC(\tau,g)\eqdef \big\langle f,g\circ \mathcal{S}_\tau h\big\rangle,
\end{equation}
where $(\mathcal{S}_\tau h)(x)=h(x-\tau)$ denotes the translation operator.

Exhaustively discretizing the six-dimensional pose space defined in Problem \eqref{eq:6d_problem} is computationally prohibitive.
A common approach is to alternate between optimization over shifts and over the rotations, e.g. \cite{chen_fast_2013}.
For fixed shift $\tau$, optimization over $g$ can be performed efficiently using Algorithm \ref{algo:coarse-to-fine}.
For fixed rotation $g$, the optimal translation can be estimated efficiently on a fine grid by evaluating the cross-correlation with 3D FFT.

More precisely, we start from zero shift and estimate the rotation using Algorithm \ref{algo:coarse-to-fine}.
Then, given the estimated rotation, we estimate the translation of the rotated volume via a 3D FFT-based correlation.
This alternating process is repeated for a small number of iterations, e.g. $T=3$ in the experiment of Section \ref{sec:exp_STA}.

\paragraph{Remarks.}
(i) Each translation update costs $\Oc (N^3\log N)$ on an $N\times N\times N$ grid (or less if restricted to a bounded search window), and is typically cheap compared to the rotation refinement at higher angular cutoffs.
(ii) The outer loop is a block coordinate ascent method and is not guaranteed to find the global maximizer of the 6-dimensional objective; nevertheless, empirically, it has been employed successfully and is known to be robust if relatively accurate initial translation is provided.
(iii) If desired, multi-resolution translation refinement can be employed to achieve sub-pixel accuracy, for example using the efficient upsampled phase-correlation scheme of \cite{guizar-sicairos_efficient_2008}.
(iv) Translational search can be restricted to a local neighborhood if an accurate initial translation is available; as is the case for the hierarchical binning refinement data typically employed in STA, see Section~\ref{sec:exp_STA}.  

\bibliographystyle{unsrt}
\bibliography{refs}

\end{document}

%% file: preambule.tex
\usepackage{amsmath}

\usepackage{amsthm}
\usepackage{amssymb}
\usepackage{amsfonts}
\usepackage{dsfont}
\usepackage{accents}
\usepackage{bm}
\usepackage{hyperref}
\usepackage{bbm}
\usepackage{mathtools}

\usepackage{tabularx} % Better tables
\usepackage{caption}
\usepackage{multirow}
\usepackage{booktabs}
\usepackage{graphicx}
\graphicspath{{Figures/}}
 \usepackage{subcaption}
 \usepackage{tcolorbox}
 \usepackage{wasysym}
  \usepackage{longtable}
\usepackage{tikz,pgfplots}
\pgfplotsset{compat=newest}
\usepgfplotslibrary{groupplots}
\usetikzlibrary{intersections,calc,arrows,matrix,spy}
\usepackage{color}
\definecolor{darkred}{rgb}{0.6,0,0}
\definecolor{darkgreen}{rgb}{0.19,0.44,0.09}
\definecolor{darkblue}{rgb}{0,0,0.5}
\definecolor{SkyBlue}{rgb}{0.53, 0.81, 0.92}
\usepackage{algorithm}
\usepackage{algpseudocode}
\usepackage[squaren,Gray]{SIunits}
\usepackage[normalem]{ulem}
\usepackage{enumitem}

\usepackage{pgfplotstable}
\usepackage{pgfplots}
\pgfplotsset{compat=1.18}
\usepgfplotslibrary{groupplots}
\pgfplotstableset{col sep=comma}

\usepackage{datatool}
\newcommand{\SE}{\mathrm{SE}}

%% file: commands.tex
\def\argmax{\mathop{\mathrm{argmax}}}

\def\R{\mathbb{R}}          									 	          % Reals
\def\C{\mathbb{C}}          									 	          % Reals
\def\SO{\mathrm{SO}}
\newcommand{\Bc}{\mathcal{B}}

\newcommand{\Gc}{\mathcal{G}}

\newcommand{\Oc}{\mathcal{O}}

\usepackage{tcolorbox}
\usepackage{mdframed}
\usepackage{lipsum}

\newtheorem{lemma}{Lemma}[section]
\newtheorem{theorem}[lemma]{Theorem}
\newtheorem{proposition}[lemma]{Proposition}
\newtheorem{corollary}[lemma]{Corollary}
\newtheorem{definition}{Definition}[section]

\theoremstyle{definition}
\newtheorem{assumption}{Assumption}[section]
\newcommand{\eqdef}{\ensuremath{\stackrel{\mbox{\upshape\tiny def.}}{=}}}

\def\rank{\mathrm{rank}}

\def\1{1\!\!1}

\def\Lso3{L^2(SO(3))}
\def\so3{SO(3)}
\def\CC{\mathrm{C}}

\newcommand{\norm}[1]{\left\lVert#1\right\rVert}

\newcommand{\PlotCurveWithLast}[7]{%
  \DTLloaddb{tmpdb}{#1}%
  \def\lasty{0}%
  \DTLforeach{tmpdb}{\yy=y}{\xdef\lasty{\yy}}%

  \pgfmathparse{round(\lasty*100)/100}%
  \edef\lastyround{\pgfmathresult}%

  \pgfmathparse{max(0.7, min(2.0, 1.0 + 0.05*\lastyround))}%
  \edef\markscale{\pgfmathresult}%

 \addplot+[
  mark=*,
  forget plot,
  color=#2,
  #5,
  mark options={scale=\markscale}
]
table[col sep=comma, x=x, y=y] {#1}
node[pos=#7, left, #6] {\scriptsize #4};

  \DTLcleardb{tmpdb}\DTLdeletedb{tmpdb}%
}

\newcommand{\PlotCurveWithLabels}[9]{%
  \pgfplotstableread[col sep=comma]{#1}\PlotTable%
  \pgfplotstablegetrowsof{\PlotTable}%
  \pgfmathtruncatemacro{\LastRow}{\pgfplotsretval-1}%
  \pgfplotstablegetelem{0}{y}\of{\PlotTable}\edef\FirstYRaw{\pgfplotsretval}%
  \pgfplotstablegetelem{\LastRow}{y}\of{\PlotTable}\edef\LastYRaw{\pgfplotsretval}%
  \pgfmathprintnumberto[fixed,precision=2]{\FirstYRaw}{\FirstYText}%
  \pgfmathprintnumberto[fixed,precision=2]{\LastYRaw}{\LastYText}%

  \edef\PlotCurveTemp{%
    \noexpand\addplot+[
      mark=*,
      color=#2,
      #4,
      mark options={draw=#2, fill=#2},
      forget plot
    ]
    \noexpand table[col sep=comma, x=x, y=y] {#1}
    \noexpand node[pos=#6, left, #5]{\noexpand\textnormal{\noexpand\scriptsize \unexpanded{#3}}}
    \noexpand node[pos=0, #7]{\noexpand\scriptsize \FirstYText$^\circ$}
    \noexpand node[pos=1, #7, xshift=#8, yshift=#9]{\noexpand\scriptsize \LastYText$^\circ$};%
  }%
  \PlotCurveTemp%
}

\newcommand{\PlotCurveWithLabelsSecond}[6]{%
  \pgfplotstableread[col sep=comma]{#1}\PlotTable%
  \pgfplotstablegetrowsof{\PlotTable}%
  \pgfmathtruncatemacro{\LastRow}{\pgfplotsretval-1}%
  \pgfplotstablegetelem{0}{y}\of{\PlotTable}\edef\FirstYRaw{\pgfplotsretval}%
  \pgfplotstablegetelem{\LastRow}{y}\of{\PlotTable}\edef\LastYRaw{\pgfplotsretval}%
  \pgfmathprintnumberto[fixed,precision=2]{\FirstYRaw}{\FirstYText}%
  \pgfmathprintnumberto[fixed,precision=2]{\LastYRaw}{\LastYText}%

  \edef\PlotCurveTemp{%
    \noexpand\addplot+[
      mark=*,
      color=#2,
      #4,
      mark options={draw=#2, fill=#2},
      forget plot
    ]
    \noexpand table[col sep=comma, x=x, y=y] {#1}
    \noexpand node[pos=#6, left, #5]{\noexpand\textnormal{\noexpand\scriptsize \unexpanded{#3}}};
  }%
  \PlotCurveTemp%
}

\newcommand{\DrawSizedLabeledPoints}[3]{%
  \pgfplotstableread[col sep=comma]{#1}\datatable
  \pgfplotstableforeachrow{\datatable}{%
    \pgfplotstablegetelem{\pgfplotstablerow}{x}\of{\datatable}\let\x\pgfplotsretval
    \pgfplotstablegetelem{\pgfplotstablerow}{y}\of{\datatable}\let\y\pgfplotsretval
    \pgfplotstablegetelem{\pgfplotstablerow}{msize}\of{\datatable}\let\r\pgfplotsretval
    \pgfplotstablegetelem{\pgfplotstablerow}{ptlabel}\of{\datatable}\let\lbl\pgfplotsretval
    \filldraw[draw=#2, fill=#2] (axis cs:\x,\y) circle[radius=\r pt];
    \node[font=\small, #3] at (axis cs:\x,\y) {\lbl};
  }%
}

\definecolor{cOv1}{HTML}{9F8FFA}
\definecolor{cOv2}{HTML}{7963F8}
\definecolor{cOv3}{HTML}{5437F6}
\definecolor{cOv5}{HTML}{2E0BF4}
\definecolor{cOv8}{HTML}{2209B5}
\definecolor{cOv10}{HTML}{1E079C}
\definecolor{cOv18}{HTML}{0A021F}
\definecolor{cOurs}{HTML}{FFB066}

\definecolor{cOv20}{HTML}{9ECAE1}
\definecolor{cOv30}{HTML}{C6DBEF}
\definecolor{cOurs10}{HTML}{FFB066}
\definecolor{cOurs100}{HTML}{CC7A00}
\definecolor{cGD}{HTML}{3A7F6F}